\documentclass[11pt]{article}
\usepackage{fullpage}
\usepackage{datetime}
\settimeformat{ampmtime}
\usdate

\usepackage{amsmath, amsthm, amssymb, bm}
\usepackage{cite}
\usepackage{graphicx}
\usepackage{epstopdf}
\usepackage{multirow}
\usepackage{url}
\usepackage{hyperref, color} 
\usepackage{subfigure}

\renewcommand{\P}[1]{\operatorname{P}\left\{#1\right\}}
\newcommand{\E}{\operatorname{E}}
\newcommand{\Tr}[1]{\operatorname{trace}\left(#1\right)}

\newcommand{\er}{\mathrm{e}}
\renewcommand{\j}{\mathrm{j}}
\newcommand{\vct}[1]{\bm{#1}}
\newcommand{\mtx}[1]{\bm{#1}}

\newcommand{\R}{\mathbb{R}}
\newcommand{\C}{\mathbb{C}}
\newcommand{\xx}{\vct{x}}
\newcommand{\ww}{\vct{w}}
\newcommand{\Bmat}{\mtx{B}}
\newcommand{\Cmat}{\mtx{C}}
\newcommand{\mm}{\vct{m}}
\newcommand{\hh}{\vct{h}}
\newcommand{\yy}{\vct{y}}
\newcommand{\yh}{\hat{\yy}}

\newcommand{\Ak}{\mtx{A}_k}
\newcommand{\bk}{\hat{\vct{b}}_k}
\newcommand{\ck}{\hat{\vct{c}}_k}
\newcommand{\I}{\mtx{I}}
\newcommand{\W}{\mtx{W}}
\newcommand{\Zk}{\mtx{Z}_k}
\newcommand{\vk}{\vct{v}_k}
\newcommand{\uk}{\vct{u}_k}
\newcommand{\uu}{\vct{u}}
\newcommand{\vv}{\vct{v}}
\newcommand{\cZ}{\mathcal{Z}}

\def\PT{\mathcal{P}_T}
\def\PTc{\mathcal{P}_{T^\perp}}
\def\cA{\mathcal{A}}
\newcommand{\<}{\langle}
\renewcommand{\>}{\rangle}

\newtheorem{lem}{Lemma}
\newtheorem{thm}{Theorem}
\newtheorem{prop}{Proposition}
\newtheorem{defn}{Definition}
\newtheorem{cor}{Corollary}
\newcommand{\revise}[1]{#1}

\begin{document}

\title{Blind Deconvolution using Convex Programming}

\author{Ali Ahmed, Benjamin Recht, and Justin Romberg\thanks{A.\ A.\, and J.\ R.\ are with the School of Electrical and Computer Engineering at Georgia Tech in Atlanta, GA.  Emails: alikhan@gatech.edu and jrom@ece.gatech.edu. B.\ R.\ is with the Department of Computer Science at the University of Wisconsin.  Email: brecht@cs.wisc.edu.  A.\ A.\ and J.\ R.\ are supported by ONR grant N00014-11-1-0459 and a grant from the Packard Foundation.  B.\ R.\  is generously supported by ONR award N00014-11-1-0723 and NSF awards CCF-1139953 and CCF-1148243. A.A. and J.R. would like to thank William Mantzel and M. Salman Asif for discussions related to this paper.  \revise{This is a corrected version of the original paper which appeared in the {\em IEEE Transactions on Information Theory}, vol.\ 60, no.\ 3, 2014.  We would also like to thank Shuyang Ling and Thomas Strohmer for pointing out the error in the original paper.} }
}

\date{\today~ (Revised Version)}

\maketitle

\begin{abstract}
	We consider the problem of recovering two unknown vectors, $\ww$ and $\xx$, of length $L$ from their circular convolution.  We make the structural assumption that the two vectors are members of known subspaces, one with dimension $N$ and the other with dimension $K$.  Although the observed convolution is nonlinear in both $\ww$ and $\xx$, it is linear in the rank-1 matrix formed by their outer product $\ww\xx^*$.  This observation allows us to recast the deconvolution problem as low-rank matrix recovery problem from linear measurements, whose natural convex relaxation is a nuclear norm minimization program.
	
	We prove the effectiveness of this relaxation by showing that for ``generic'' signals, the program can deconvolve $\ww$ and $\xx$ exactly when the maximum of $N$ and $K$ is almost on the order of $L$.  That is, we show that if $\xx$ is drawn from a random subspace of dimension $N$, and $\ww$ is a vector in a subspace of dimension $K$ whose basis vectors are ``spread out'' in the frequency domain, then nuclear norm minimization recovers $\ww\xx^*$ without error.
	
	We discuss this result in the context of blind channel estimation in communications.  If we have a message of length $N$ which we code using a random $L\times N$ coding matrix, and the encoded message travels through an unknown linear time-invariant channel of maximum length $K$, then the receiver can recover both the channel response and the message when $L\gtrsim N+K$, to within constant and log factors.
\end{abstract}
\section{Introduction}

This paper considers a fundamental problem in signal processing and communications: we observe the convolution of two unknown signals, $\ww$ and $\xx$, and want to separate them.  We will show that this problem can be naturally relaxed as a semidefinite program (SDP), in particular, a nuclear norm minimization program.  We then use this fact in conjunction with recent results on recovering low-rank matrices from underdetermined linear observations to provide conditions under which $\ww$ and $\xx$ can be deconvolved exactly.  Qualitatively, these results say that if both $\ww$ and $\xx$ have length $L$, $\ww$ lives in a fixed subspace of dimension $K$ and is spread out in the frequency domain, and $\xx$ lives in a ``generic'' subspace chosen at random, then $\ww$ and $\xx$ are separable with high probability.

The general statement of the problem is as follows.  We will assume that the length $L$ signals live in known subspaces of $\R^L$ whose dimensions are $K$ and $N$.  That is, we can write
\begin{align*}
	\ww &= \Bmat\hh, \quad \hh\in\R^K\\
	\xx &= \Cmat\mm, \quad \mm\in\R^N
\end{align*}
for some $L\times K$ matrix $\Bmat$ and $L\times N$ matrix $\Cmat$.  The columns of these matrices provide bases for the subspaces in which $\ww$ and $\xx$ live; recovering $\hh$ and $\mm$, then, is equivalent to recovering $\ww$ and $\xx$.  

We observe the circular
convolution of $\ww$ and $\xx$:
\begin{equation}
	\label{eq:yconv}
	\yy = \ww \ast \xx,
	\quad\text{or}\quad
	y[\ell] = \sum_{\ell'=1}^{L}w[\ell']x[\ell-\ell'+1], 
\end{equation}
where the index $\ell-\ell'+1$ in the sum above is understood to be modulo $\{1,\ldots,L\}$.  It is clear that without structural assumptions on $\ww$ and $\xx$, there will not be a unique separation given the observations $\yy$.  But we will see that once we account for our knowledge that $\ww$ and $\xx$ lie in the span of the columns of $\Bmat$ and $\Cmat$, respectively, they can be uniquely separated in many situations.  This paper details a particular set of conditions under which the signals can be deconvolved tractably via convex programming.

\subsection{Notations}
Unless specified otherwise, we use uppercase bold, lowercase bold, and not bold letters for matrices, vectors, and scalars, respectively. For example, $\mtx{X}$ denotes a matrix, $\vct{x}$ represents a vector, and $x$ refers to a scalar. Calligraphic letters such as $\cA$ specify linear operators. The symbol $C$ refers to a constant number, which may not refer to the same number every time it is used. The notations $\|\cdot\|$, $\|\cdot\|_*$, and $\|\cdot\|_F$ denote the operator, nuclear, and Frobenius norms of the matrices, respectively. Furthermore, we will use $\|\cdot\|_2$, and $\|\cdot\|_1$ to represent the vector $\ell_2$, and $\ell_1$ norms.
\subsection{Matrix observations}
\label{sec:matrixobs}

We can break apart the convolution in \eqref{eq:yconv} by expanding $\xx$ as a linear combination of the columns $\Cmat_1,\ldots,\Cmat_N$ of $\Cmat$,
\begin{align*}
	\yy &= m(1)\ww\ast\Cmat_1 + m(2)\ww\ast\Cmat_2 + \cdots + m(N)\ww\ast\Cmat_N \\
	&= 
	\begin{bmatrix} 
		\operatorname{circ}(\Cmat_1) & \operatorname{circ}(\Cmat_2) & \cdots & \operatorname{circ}(\Cmat_N)
	\end{bmatrix}
	\begin{bmatrix}
		m(1)\ww \\ m(2)\ww \\ \vdots \\ m(N)\ww
	\end{bmatrix},
\end{align*}
where $\operatorname{circ}(\Cmat_n)$ corresponds to the $L\times L$ circulant matrix whose action corresponds to circular convolution with the vector $\Cmat_n$.  Expanding $\ww$ as a linear combination of the columns of $\Bmat$, this becomes
\begin{align}
	\label{eq:ymattime}
	\yy &= 
	\begin{bmatrix} 
		\operatorname{circ}(\Cmat_1)\Bmat & \operatorname{circ}(\Cmat_2)\Bmat & \cdots &
		\operatorname{circ}(\Cmat_N)\Bmat
	\end{bmatrix}
	\begin{bmatrix}
		m(1)\hh \\ m(2)\hh \\ \vdots \\ m(N)\hh
	\end{bmatrix}.
\end{align}

We will find it convenient to write \eqref{eq:ymattime} in the Fourier domain.  Let $\mtx{F}$ be the $L$-point normalized discrete Fourier transform (DFT) matrix
\[
	\mtx{F}(\omega,\ell) = \frac{1}{\sqrt{L}}\er^{-\j2\pi(\omega-1)(\ell-1)/L},
	\qquad 1\leq\omega,\ell\leq L.
\]
We will use $\hat{\Cmat}=\mtx{F}\Cmat$ for the $\Cmat$-basis transformed into the Fourier domain, and also $\hat{\Bmat} = \mtx{F}\Bmat$.  Then $\operatorname{circ}(\Cmat_n) = \mtx{F}^*\Delta_n\mtx{F}$, where $\Delta_n$ is a diagonal matrix constructed from the $n$th column of $\hat{\Cmat}$, $\Delta_n = \operatorname{diag}(\sqrt{L}\hat{\Cmat}_n)$, and \eqref{eq:ymattime} becomes
\begin{align}
	\label{eq:ymatfreq}
	\yh = \mtx{F}\yy &=
	\begin{bmatrix} 
		\Delta_1\hat{\Bmat} & \Delta_2\hat{\Bmat} & \cdots &
		\Delta_N\hat{\Bmat}
	\end{bmatrix}
	\begin{bmatrix}
		m(1)\hh \\ m(2)\hh \\ \vdots \\ m(N)\hh
	\end{bmatrix}.
\end{align}
Clearly, recovering $\vct{m}$ and $\vct{h}$ from $\yh$ is the same as recovering $\vct{x}$ and $\vct{w}$ from $\yy$.

The expansions \eqref{eq:ymattime} and \eqref{eq:ymatfreq} make it clear that while $\yy$ is a nonlinear combination of the coefficients $\hh$ and $\mm$, it is a {\em linear} combination of the entries of their outer product $\mtx{X}_0 = \hh\mm^*$.  We can pose the blind deconvolution problem as a linear inverse problem where we want to recover a $K\times N$ matrix from observations
\begin{equation}
	\label{eq:Adef}
	\yh = \cA(\mtx{X}_0),
\end{equation}
through a linear operator $\cA$ which maps $K\times N$ matrices to $\R^L$.  For $\cA$ to be invertible over all matrices, we need at least as many observations as unknowns, $L\geq NK$.  But since we know $\mtx{X}_0$ has special structure, namely that its rank is 1, we will be able to recover it from $L\ll NK$ under certain conditions on $\cA$.

As each entry of $\yh$ is a linear combination of the entries in $\hh\mm^*$, we can write them as trace inner products of different $K\times N$ matrices against $\hh\mm^*$.  Using $\hat{\vct{b}}_\ell\in\C^K$ for the $\ell$th column of $\hat{\Bmat}^*$ and $\hat{\vct{c}}_\ell\in\C^N$ as the $\ell$th row of $\sqrt{L}\hat{\Cmat}$, we can translate one entry in \eqref{eq:ymatfreq} as\footnote{As we are now manipulating complex numbers in the frequency domain, we will need to take a little bit of care with definitions.  Here and below, we use $\<\vct{u},\vct{v}\> = \vct{v}^*\vct{u} = \Tr{\vct{u}\vct{v}^*}$ for complex vectors $\vct{u}$ and $\vct{v}$.}
\begin{align}
	\nonumber
	\hat{y}(\ell) &= \hat{c}_\ell(1)m(1)\<\hh,\hat{\vct{b}}_\ell\> + \hat{c}_\ell(2)m(2)\<\hh,\hat{\vct{b}}_\ell\> + \cdots +
	\hat{c}_\ell(N)m(N)\<\hh,\hat{\vct{b}}_\ell\> \\
	\nonumber
	&= \<\hat{\vct{c}}_\ell,\mm\>\<\hh,\hat{\vct{b}}_\ell\> \\
	&= 
	\label{eq:Al}
	\Tr{\mtx{A}_\ell^*(\hh\mm^*)},\quad\text{where}\quad \mtx{A}_\ell = \hat{\vct{b}}_\ell\hat{\vct{c}}_\ell^*.
\end{align}

Now that we have seen that separating two signals given their convolution can be recast as a matrix recovery problem, we turn our attention to a method for solving it.  In the next section, we argue that a natural way to recover the expansion coefficients $\mm$ and $\hh$ from measurements of the form \eqref{eq:ymatfreq} is using nuclear norm minimization.

\subsection{Convex relaxation}

The previous section demonstrated how the blind deconvolution problem can be recast as a linear inverse problem over the (nonconvex) set of rank-$1$ matrices.  A common heuristic to convexify the problem is to use the {\em nuclear norm}, the sum of the singular values of a matrix, as a proxy for rank \cite{fazel02ma}.  In this section, we show how this heuristic provides a natural convex relaxation. 

Given $\yh\in\C^L$, our goal is to find $\hh\in\R^K$ and $\mm\in\R^N$ that are consistent with the observations in \eqref{eq:ymatfreq}.  Making no assumptions about either of these vectors other than the dimension, the natural way to choose between multiple feasible points is using least-squares.  We want to solve
\begin{equation}
	\label{eq:least-squares-opt}
	\min_{\vct{u},\vct{v}}~\|\vct{u}\|_2^2 + \|\vct{v}\|_2^2 \quad 
	\mbox{subject to}\quad 
	\yh(\ell) = \<\hat{\vct{c}}_\ell,\vct{u}\>\<\vct{v},\hat{\vct{b}}_\ell\>,
	\quad \ell=1,\ldots,L.
\end{equation}
This is a non-convex quadratic optimization problem.  The cost function is convex, but the quadratic equality constraints mean that the feasible set is non-convex.  A standard approach to solving such quadratically constrained quadratic programs is to use duality (see for example \cite{Nesterov00}).  A standard calculation shows that the dual of \eqref{eq:least-squares-opt} is the semi-definite program (SDP)
\begin{align}
	\label{eq:sdp-dual}
	&\max_{\vct{\lambda}} \quad \operatorname{Re}\<\hat{\vct{y}},\vct{\lambda}\> \\
	\nonumber
	&\quad \text{subject to} \quad
	\begin{bmatrix}
		\mtx{I} & \sum_{\ell=1}^L \lambda(\ell) \mtx{A}_\ell\\
		\sum_{\ell=1}^L \lambda(\ell)^* \mtx{A}_\ell^* & \mtx{I}
	\end{bmatrix}
	\succeq 0,
\end{align}
with the $\mtx{A}_\ell=\hat{\vct{b}}_\ell\hat{\vct{c}}_\ell^*$ defined as in the previous section.  Taking the dual again will give us a convex program which is in some sense as close to \eqref{eq:least-squares-opt} as possible.  The dual SDP of \eqref{eq:sdp-dual} is \cite{recht10gu}
\begin{align}
	\label{eq:sdp-dual-dual}
	&\min_{\mtx{W}_1,\mtx{W}_2,\mtx{X}}~
	 \tfrac{1}{2} \Tr{\mtx{W}_1} + \tfrac{1}{2} \Tr{\mtx{W}_2} \\\nonumber
	&\quad \text{subject to} \quad
	\begin{bmatrix}
		\mtx{W}_1 &  \mtx{X}\\
		\mtx{X}^* & \mtx{W}_2
	\end{bmatrix}
	\succeq 0 \\ \nonumber
	& \qquad\qquad\qquad~~~\hat{\vct{y}} = \cA(\mtx{X}),
\end{align} 
which is equivalent to 
\begin{equation}
	\label{eq:nuclear-opt}
	\begin{array}{ll} \mbox{min} & \|\mtx{X}\|_*\\
	\mbox{subject to}
	& \hat{\vct{y}} = \cA(\mtx{X})
	\end{array}.
\end{equation}
That is, the nuclear norm heuristic is the ``dual-dual'' relaxation of the intuitive but non-convex least-squares estimation problem \eqref{eq:least-squares-opt}. 

Our technique for untangling $\ww$ and $\xx$ from their convolution, then, is to take the Fourier transform of the observation $\yy=\ww\ast\xx$ and use it as constraints in the program \eqref{eq:nuclear-opt}.  That \eqref{eq:nuclear-opt} is the natural relaxation is fortunate, as an entire body of literature in the field of {\em low-rank recovery} has arisen in the past five years that is devoted to analyzing problems of the form \eqref{eq:nuclear-opt}.  
We will build on some of the techniques from this area in establishing the theoretical guarantees for when \eqref{eq:nuclear-opt} is provably effective presented in the next section.

There have also been tremendous advances in algorithms for computing the solution to optimization problems of both types \eqref{eq:least-squares-opt} and \eqref{eq:nuclear-opt}.  In Section~\ref{sec:largescale}, we will briefly detail one such technique we used to solve \eqref{eq:least-squares-opt} on a relatively large scale for a series of numerical experiments in Sections~\ref{sec:phase}--\ref{sec:imagedeblur}.
%
%
\subsection{Main results}
\label{sec:mainresult}

We can guarantee the effectiveness of \eqref{eq:nuclear-opt} for relatively large subspace dimensions $K$ and $N$ when $\Bmat$ is incoherent in the Fourier domain, and when $\Cmat$ is generic.  Before presenting our main analytical result, Theorem~\ref{th:main} below, we will carefully specify our models for $\Bmat$ and $\Cmat$, giving a concrete definition to the terms `incoherent' and `generic' in the process.

\revise{
We will assume that the signal $\ww$ is time-limited to $Q$, where $K\leq Q \leq L$.  This means that the last $L-Q$ rows of $\Bmat$ are zero (see \eqref{eq:Bembed} below, for example).  We will also assume, without additional loss of generality, that the columns of $\Bmat$ are orthonormal:
}
\begin{equation}
	\label{eq:Biso}
	\Bmat^*\Bmat = \hat{\Bmat}^*\hat{\Bmat} = \sum_{\ell=1}^L \hat{\vct{b}}_\ell\hat{\vct{b}}_\ell^* = \I,
\end{equation}
where the $\hat{\vct{b}}_\ell$ are the columns of $\hat{\Bmat}^*$, as in \eqref{eq:Al}.  Our results will be most powerful when $\Bmat$ is diffuse in the Fourier domain, meaning that the $\hat{\vct{b}}_\ell$ all have similar norms.  We will use the (in)coherence parameter $\mu_{\max}$ to quantify the degree to which the columns of $\Bmat$ are jointly concentrated in the Fourier domain:
\begin{equation}
	\label{eq:mu1def}
	\mu_{\max}^2 = \frac{L}{K}\max_{1\leq\ell\leq L} \|\hat{\vct{b}}_\ell\|^2_2.
\end{equation}
From \eqref{eq:Biso}, we know that the total energy in the rows of $\hat{\Bmat}$ is $\sum_{\ell=1}^L\|\hat{\vct{b}}_\ell\|^2_2 = K$, and that $\|\hat{\vct{b}}_\ell\|^2_2\leq 1$.  Thus $1\leq\mu_{\max}^2\leq L/K$, with the coherence taking its minimum value when the energy in $\hat{\Bmat}$ is evenly distributed throughout its rows, and its maximum value when the energy is completely concentrated on $K$ of the $L$ rows.  Our results will also depend on the minimum of these norms
\begin{equation}
	\label{eq:muLdef}
	\mu_{\min}^2 = \frac{L}{K}\min_{1\leq\ell\leq L} \|\hat{\vct{b}}_\ell\|^2_2.
\end{equation}
We will always have $0\leq\mu_{\min}^2\leq 1$ and $\mu_{\min}^2\leq\mu_{\max}^2$.  An example of a maximally incoherent $\Bmat$, where $\mu_{\max}^2=\mu_{\min}^2=1$, is 
\begin{equation}
	\label{eq:Bembed}
	\Bmat = 
	\begin{bmatrix}
		\I_K \\ \mtx{0}
	\end{bmatrix},
\end{equation}
where $\I_K$ is the $K\times K$ identity matrix.  In this case, the range of $\Bmat$ consists of ``short'' signals whose first $K$ terms may be non-zero.  The matrix $\hat{\Bmat}$ is simply the first $K$ columns of the discrete Fourier matrix, and so every entry has the same magnitude.

Our analytic results also depend on how diffuse the particular signal we are trying to recover $\ww = \Bmat\hh$ is in the Fourier domain.  With $\hat{\ww} = \mtx{F}\ww = \hat{\Bmat}\hh$, we define
\begin{equation}
	\label{eq:muhdef}
	\mu_h^2 = L\max_{1\leq\ell\leq L} |\hat{w}(\ell)|^2 = 
	 L\cdot\max_{1\leq\ell\leq L} |\<\vct{h},\hat{\vct{b}}_\ell\>|^2.
\end{equation}
Note that it is always the case that $1\leq\mu_h^2\leq \mu_{\max}^2K$. The lower bound follows from the Cauchy-Schwarz inequality, i.e.,
\[
\mu_h^2 \leq L\cdot\max_{1\leq\ell\leq L}\|\hat{\vct{b}}_\ell\|_2^2\|\hh\|_2^2 \leq K \mu_{\max}^2,
\]
where the last inequality is the result of \eqref{eq:mu1def}, and $\|\hh\|_2 = 1$. To show the lower bound of $\mu_h^2$, take a summation over $\ell$ on both sides
\[
\sum_{\ell = 1}^L \mu_h^2 = L\sum_{\ell = 1}^L\max_{1\leq\ell\leq L} |\<\vct{h},\hat{\vct{b}}_\ell\>|^2,
\]
which means
\begin{align*}
\mu_h^2 &\geq \sum_{\ell = 1}^L|\<\vct{h},\hat{\vct{b}}_\ell\>|^2 = \|\hh\|_2^2 = 1,
\end{align*}
where the equality holds because $\hat{\mtx{B}}$ is a matrix with orthonormal columns. As an illustration, if $\mtx{B}$ is as in \eqref{eq:Bembed} (i.e., $\hat{\mtx{B}}$ is the partial Fourier matrix), then $\mu_h^2$ quantifies the dispersion of $\ww$ in the frequency domain. In particular, if the signal $\ww$ is more or less ``flat'' in the frequency domain, then $\mu_h^2$ will be a small constant. 
%

With the subspace in which $\ww$ resides fixed, we will show that separating $\ww$ and $\xx=\Cmat\mm$ will be possible for ``most'' choices of the subspace $\Cmat$ of a certain dimension $N$ --- we do this by choosing the subspace at random from an isotropic distribution, and show that \eqref{eq:nuclear-opt} is successful with high probability.  For the remainder of the paper, we will take the entries of $\Cmat$ to be independent and identically distributed random variables,
\[
	C[\ell,n] ~\sim~ \mathrm{Normal}(0,L^{-1}).
\]
In the Fourier domain, the entries of $\hat{\Cmat}$ will be complex Gaussian, and its columns will have conjugate symmetry (since the columns of $\Cmat$ are real).  Specifically, the rows of $\hat{\Cmat}$ will be distributed as\footnote{We are assuming here that $L$ is even; the argument is straightforward to adapt to odd $L$.} 
\begin{align}
	\label{eq:normalck}
	\hat{\vct{c}}_\ell &\sim
	\begin{cases}
		\mathrm{Normal}(0, \I) & \ell = 1 \\
		\mathrm{Normal}(0,2^{-1/2}\I) + \j\mathrm{Normal}(0,2^{-1/2}\I) & \ell = 2,\ldots,L/2+1
	\end{cases}, \\
	\nonumber
	\hat{\vct{c}}_\ell &= \hat{\vct{c}}_{L-\ell+2},\quad\text{for}~~\ell=L/2+2,\ldots,L.	
\end{align}
Similar results to those we present here most likely hold for other models for $\Cmat$.  
The key property that our analysis hinges critically on is the rows $\hat{\vct{c}}_\ell$ of $\hat{\Cmat}$ are {\em independent} --- this allows us to apply recently developed tools for estimating the spectral norm of a sum of independent random linear operators.  

We now state our main result:
\begin{thm}
	\label{th:main}
	\revise{Fix $\alpha\geq 1$.  Let $\mtx{B}$ be a deterministic $L\times K$ matrix satisfying \eqref{eq:Biso} and whose last $L-Q$ rows are zero with\footnote{Throughout the manuscript we will use the notation $C_\alpha$ to denote a constant which depends only on the probability exponent $\alpha$.  Its value may be different from instantiation to instantiation.}
	\[
		Q~\geq~C_\alpha\cdot M\log(L)\log M,\quad M = \max(\mu_{\mathrm{max}}^2K,\mu_h^2N),
	\]
	and $L$ as an integer multiple of $Q$ with $L/Q\geq \log(C_\alpha'\sqrt{N\log L})/\log 2$.  Let $\mtx{C}$ be an $L\times N$ Gaussian random matrix drawn as in \eqref{eq:normalck}, and set $\vct{w}=\mtx{B}\vct{h}$ and $\vct{x}=\mtx{C}\vct{m}$ for arbitrary basis coefficients $\vct{h}\in\R^K$, $\vct{m}\in\R^N$.  Let the coherence parameters of the basis $\mtx{B}$, and expansion coefficients $\vct{h}$ be as defined in   \eqref{eq:mu1def}, and \eqref{eq:muhdef} above.    Then there exists a constant $C_\alpha''=O(\alpha)$ depending only on $\alpha$, such that if
	\begin{equation}
		\label{eq:Mbound}
		\max\left(\mu_{\max}^2K,~\mu_h^2N\right) ~\leq~
		\frac{L}{C_\alpha''\log^3 L},
	\end{equation}
	then $\mtx{X}_0=\vct{h}\vct{m}^*$ is the unique solution to \eqref{eq:nuclear-opt} with probability $1-O(L^{-\alpha+1})$, and we can recover both $\vct{w}$ and $\vct{x}$ (within a scalar multiple) from $\vct{y}=\vct{w}*\vct{x}$. }
\end{thm}

When the coherences are low, meaning that $\mu_{\max}$ and $\mu_h$ are on the order of a constant, then \eqref{eq:Mbound} is tight to within a logarithmic factor, as we always have $\max(K,N)\leq L$.

While Theorem~\ref{th:main} establishes theoretical guarantees for specific types of subspaces specified by $\Bmat$ and $\Cmat$, we have found that treating blind deconvolution as a linear inverse problem with a rank constraint leads to surprisingly good results in many situations; see, for example, the image deblurring experiments in Section~\ref{sec:imagedeblur}.

The recovery can also be made stable in the presence of noise, as described by our second theorem:
\begin{thm}
	\label{th:stability}
	Let $\mtx{X}_0=\hh\mm^*$ and $\cA$ as in \eqref{eq:Adef} with $N,K,L$ obeying \eqref{eq:Mbound}.  We observe 
	\[
		\hat{\yy} = \cA(\mtx{X}_0) + \vct{z},
	\] 
	where $\vct{z}\in\R^L$ is an unknown noise vector with $\|\vct{z}\|_2\leq\delta$, and estimate $\mtx{X}_0$ by solving
	\begin{equation}
		\label{eq:nnrelaxed}
		\begin{array}{ll} \min & \|\mtx{X}\|_*\\
		\mathrm{subject~to}
		& \|\hat{\vct{y}} - \cA(\mtx{X})\|_2\leq\delta
		\end{array}.
	\end{equation}
	Let $\lambda_{\mathrm{min}}$ be the smallest non-zero eigenvalue of $\cA\cA^*$, and $\lambda_{\mathrm{max}}$ be the largest.  Then with the same probability $1-L^{-\alpha+1}$ as in Theorem~\ref{th:main} the solution $\tilde{\mtx{X}}$ to \eqref{eq:nnrelaxed} will obey
	\begin{equation}
		\label{eq:stablerecovery}
		\|\tilde{\mtx{X}} - \mtx{X}_0\|_F ~\leq~
		C\,\frac{\lambda_{\max}}{\lambda_{\min}}\sqrt{\min(K,N)}\,\delta,
	\end{equation}
	for a fixed constant $C$.
\end{thm}
The program in \eqref{eq:nnrelaxed} is also convex, and is solved with numerical techniques similar to the equality constrained program in \eqref{eq:nuclear-opt}.  The performance bound relies on the conditioning of $\cA\cA^*$.  Lemma~\ref{lm:cAcAcond} below tells us that when $\cA$ is sufficiently underdetermined,
\begin{equation}
	\label{eq:stableover}
	NK ~\geq~ \frac{C_\alpha}{\mu^2_{\mathrm{min}}}\, L\log^2 L,
\end{equation}
then with high probability we can replace the ratio of eigenvalues in \eqref{eq:stablerecovery} with the ratio of coherence parameters for $\hat{\Bmat}$, as
\[
	\frac{\lambda_{\max}}{\lambda_{\min}} ~\sim~ \frac{\mu_{\mathrm{max}}}{\mu_{\mathrm{min}}}.  
\] 
For $L$ large enough, there will be many $N$ and $K$ which satisfy \eqref{eq:stableover} and \eqref{eq:Mbound} simultaneously.



In the end, we are interested in how well we recover $\xx$ and $\ww$.  The stability result for $\mtx{X}_0$ can easily be extended to a guarantee for the two unknown vectors.
\begin{cor}	
	Let $\tilde{\sigma_1}\tilde{\vct{u}_1}\tilde{\vct{v}_1}$ be the best rank-1 approximation to $\tilde{\mtx{X}}$, and set $\tilde{\hh}=\sqrt{\tilde{\sigma_1}}\tilde{\vct{u}_1}$ and $\tilde{\mm}=\sqrt{\tilde{\sigma_1}}\tilde{\vct{v}_1}$.  Set $\tilde{\delta}=\|\tilde{\mtx{X}} - \mtx{X}_0\|_F$.  Then there exists a constant $C$ such that
	\[
		\|\hh - \alpha\tilde{\hh}\|_2 ~\leq~ C\min\left(\tilde{\delta}/\|\hh\|_2,\|\hh\|_2\right),
		\qquad
		\|\mm - \alpha^{-1}\tilde{\mm}\|_2 ~\leq~ C\min\left(\tilde{\delta}/\|\mm\|_2,\|\mm\|_2\right).
	\]
	for some scalar multiple $\alpha$.
\end{cor}
Proof of this corollary follows the exact same line of reasoning as the later part of Theorem 1.2 in \cite{candes12ph}.
\subsection{Relationship to phase retrieval and other quadratic problems}

Blind deconvolution of $\ww\ast\xx$, as is apparent from \eqref{eq:yconv}, is equivalent to solving a system of quadratic equations in the entries of $\ww$ and $\xx$.  The discussion in Section~\ref{sec:matrixobs} shows how this system of quadratic equations can be recast as a linear set of equations with a rank constraint.  In fact, this same recasting can be used for any system of quadratic equations in $\ww$ and $\xx$.  The reason is simple: taking the outer product of the concatenation of $\ww$ and $\xx$ produces a rank-1 matrix that contains all the different combinations of entries of $\ww$ multiplied with each other and multiplied by entries in $\xx$:
\begin{equation}
	\label{eq:rank1lift}
	\begin{bmatrix} \ww \\ \xx \end{bmatrix}
	\begin{bmatrix} \ww^* & \xx^* \end{bmatrix}
	=
	\left[
	\begin{array}{cccc|cccc}
		w[1]^2 & w[1]w[2] & \cdots & w[1]w[L] & w[1]x[1] & w[1]x[2] & \cdots & w[1]x[L]\\
		w[2]w[1] & w[2]^2 & \cdots & w[2]w[L] & w[2]x[1] & w[2]x[2] & \cdots & w[2]x[L]\\
		\vdots & & & & \vdots & & \vdots \\
		w[L]w[1] & w[L]w[2] & \cdots & w[L]^2 & w[L]x[1] & w[L]x[2] & \cdots & w[L]x[L]\\ \hline
		x[1]w[1] & x[1]w[2] & \cdots & x[1]w[L] & x[1]^2 & x[1]x[2] & \cdots & x[1]x[L]\\
		x[2]w[1] & x[2]w[2] & \cdots & x[2]w[L] & x[2]x[1] & x[2]^2 & \cdots & x[2]x[L]\\
		\vdots & & & & \vdots & & \vdots \\
		x[L]w[1] & x[L]w[2] & \cdots & x[L]w[L] & x[L]x[1] & x[L]x[2] & \cdots & x[L]^2
	\end{array}
	\right].
\end{equation}
Then any quadratic equation can be written as a linear combination of the entries in this matrix, and any system of equations can be written as a linear operator acting on this matrix.  For the particular problem of blind deconvolution, we are observing sums along the skew-diagonals of the matrix in the upper right-hand (or lower left-hand) quadrant.  Incorporating the subspace constraints allows us to work with the smaller $K\times N$ matrix $\hh\mm^*$, but this could also be interpreted as adding additional linear constraints on the matrix in \eqref{eq:rank1lift}.

Recent work on {\em phase retrieval} \cite{candes12ph} has used this same methodology of ``lifting'' a quadratic problem into a linear problem with a rank constraint to show that a vector $\vct{w}\in\R^N$ can be recovered from $O(N\log N)$ measurements of the form $|\<\vct{w},\vct{a}_n\>|^2$ for $\vct{a}_n$ selected uniformly at random from the unit sphere.  In this case, the measurements are being made entirely in the upper left-hand (or lower-right hand) quadrant in \eqref{eq:rank1lift}, and the measurements in \eqref{eq:Al} have the form $\mtx{A}_n = \vct{a}_n\vct{a}_n^*$.  In fact, another way to interpret the results in \cite{candes12ph} is that if a signal of length $L$ is known to live in a generic subspace of dimension $\sim L/\log L$, then it can be recovered from an observation of a convolution with itself.  Phase retrieval using convex programming was also explored in \cite{moravec07co,chan08te}

In the current work, we are considering a non-symmetric rank-1 matrix being measured by matrices $\hat{\vct{b}}_\ell\hat{\vct{c}}_\ell^*$ formed by the outer product of two different vectors, one of which is random, and one of which is fixed.  Another way to cast the problem, which perhaps brings these differences into sharper relief, is that we are measuring the symmetric matrix in \eqref{eq:rank1lift} by taking inner products against rank-two matrices $\frac{1}{2}\left(\begin{bmatrix}\hat{\vct{b}}_\ell \\ \vct{0} \end{bmatrix}\begin{bmatrix} \vct{0} & \hat{\vct{c}}_\ell^*\end{bmatrix} + \begin{bmatrix} \vct{0} \\ \hat{\vct{c}}_\ell\end{bmatrix}\begin{bmatrix} \hat{\vct{b}}_\ell^* & \vct{0} \end{bmatrix}\right)$.  These seemingly subtle differences lead to a considerably different mathematical treatment.

\subsection{Application: Multipath channel protection using random codes}
\label{sec:channelprotect}

The results in Section~\ref{sec:mainresult} have a direct application in the context of channel coding for transmitting a message over an unknown multipath channel.  The problem is illustrated in Figure~\ref{fig:channelprotect}.  A message vector $\mm\in\R^N$ is encoded through an $L\times N$ encoding matrix $\mtx{C}$.  The protected message $\xx=\mtx{C}\mm$ travels through a channel whose impulse response is $\ww$.  The receiver observes $\vct{y} = \ww\ast\xx$, and from this would like to jointly estimate the channel and determine the message that was sent.

In this case, a reasonable model for the channel response $\ww$ is that it is nonzero in relatively small number of known locations.  Each of these entries corresponds to a different path over which the encoded message traveled; we are assuming that we know the timing delays for each of these paths, but not the fading coefficients.  The matrix $\Bmat$ in this case is a subset of columns from the identity, and the $\hat{\vct{b}}_\ell$ are partial Fourier vectors.  This means that the coherence $\mu_{\max}$ in \eqref{eq:mu1def} takes its minimal value of $\mu_{\max}^2=1$, and the coherence $\mu_h^2$ in \eqref{eq:muhdef} has a direct interpretation as the peak-value of the (normalized) frequency response of the unknown channel.  The resulting linear operator $\cA$ corresponds to a matrix comprised of $N$ $L\times K$ random Toeplitz matrices, as shown in Figure~\ref{fig:multitoeplitz}.  The first column of each of these matrices corresponds to a columns of $\Cmat$.
The formulation of this problem as a low-rank matrix recovery program was proposed in \cite{asif09ra}, which presented some first numerical experiments.

In this context, Theorem~\ref{th:main} tell us that a length $N$ message can be protected against a channel with $K$ {\em reflections} that is relatively flat in the frequency domain with a random code whose length $L$ obeys $L/\log^3 L\gtrsim (K+N)$.  Essentially, we have a theoretical guarantee that we can estimate the channel without knowledge of the message from a single transmitted codeword.  


It is instructive to draw a comparison in to previous work which connected error correction to structured solutions to underdetermined systems of equations.  In \cite{candes05de,rudelson05ge}, it was shown that a message of length $N$ could be protected against corruption in $K$ unknown locations with a code of length $L\gtrsim N + K\log(N/K)$ using a random codebook.  This result was established by showing how the decoding problem can be recast as a sparse estimation problem to which results from the field of compressed sensing can be applied.  

For multipath protection, we have a very different type of corruption: rather than individual entries of the transmitted vector being tampered with, instead we observe overlapping copies of the transmission.  We show that with the same type of codebook (i.e.\ entries chosen independently at random) can protect against $K$ reflections during transmission, where the timing of these bounces is known (or can be reasonably estimated) but the fading coefficients (amplitude and phase change associated with each reflection) are not.
%
%
\begin{figure}
	\centering
	\includegraphics[height=2.5in]{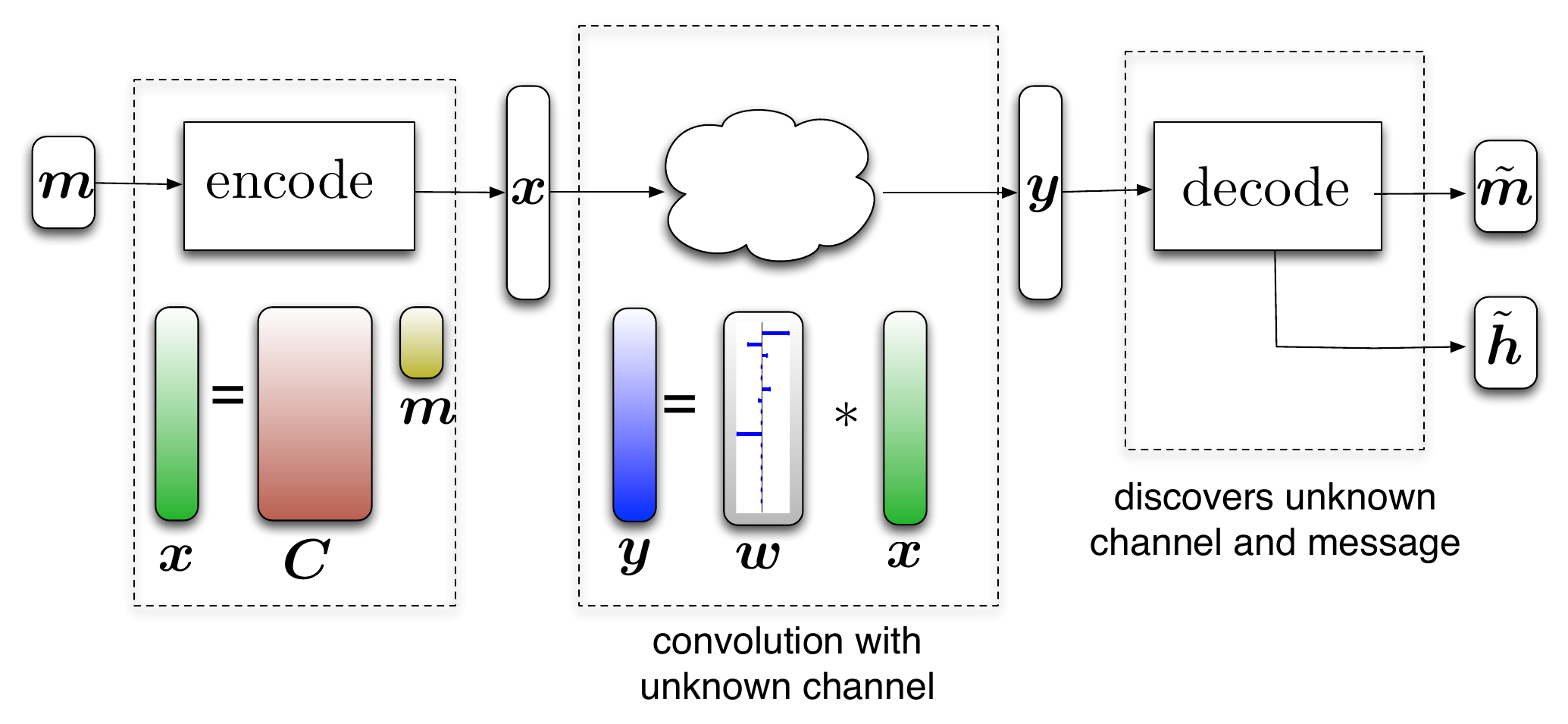}
	\caption{ Overview of the channel protection problem.  A message $\mm$ is encoded by applying a tall matrix $\mtx{C}$; the receiver observes the encoded message convolved with an unknown channel response $\ww=\Bmat\hh$, where $\Bmat$ is a subset of columns from the identity matrix.  The decoder is faced with the task of separating the message and channel response from this convolution, which is a nonlinear combination of $\hh$ and $\mm$.}
	\label{fig:channelprotect}
\end{figure}
%
\begin{figure}
	\centering
	\includegraphics[width=6in]{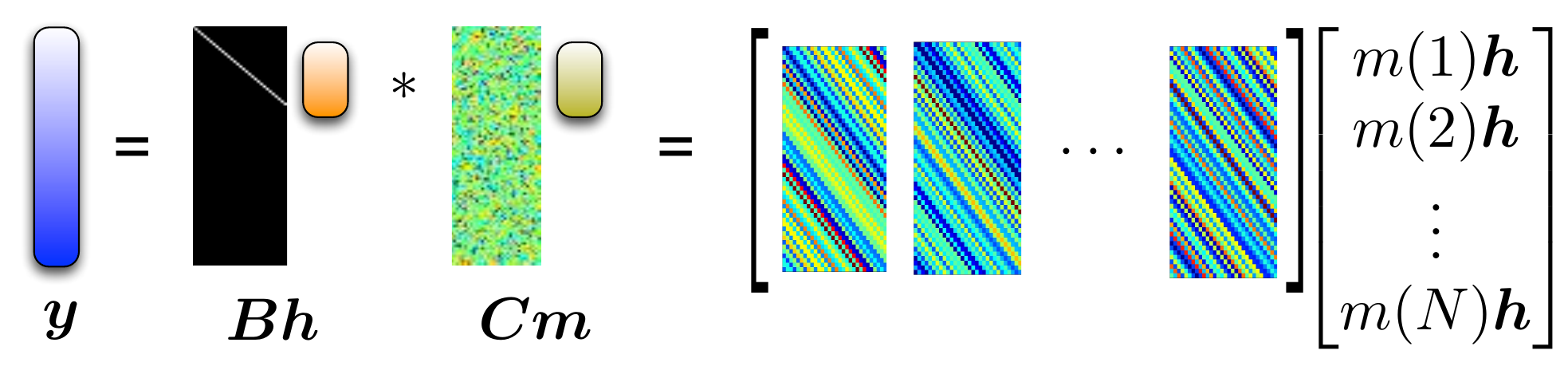}
	\caption{The multi-toeplitz matrix corresponding to the multipath channel protection problem in Section~\ref{sec:channelprotect}.  In this case, the columns of $\Bmat$ are sampled from the identity, the entries of $\Cmat$ are chosen to be iid Gaussian random variables, and the corresponding linear operator $\cA$ is formed by concatenating $N$ $L\times K$ random Toeplitz matrices, each of which is generated by a column of $\Cmat$.}
	\label{fig:multitoeplitz}
\end{figure}
%
%
%
%
\subsection{Other related work}

As it is a ubiquitous problem, many different approaches for blind deconvolution have been proposed in the past, each using different statistical or deterministic models tailored to particular applications.  
A general overview for blind deconvolution techniques in imaging (including methods based on parametric modeling of the inputs and incorporating spatial constraints) can be found in \cite{kundur96bl}.  An example of a more modern method can be found in \cite{chan98to}, where it is demonstrated how an image, which is expected to have small total-variation with respect to its energy, can be effectively deconvolved from an unknown kernel with known compact support. Since the advent of compressed sensing and our understanding of the fact that the $\ell_1$ penalty favors sparsity, several works; see, for example, \cite{saligrama09co,krishnan11bl}, imposed the $\ell_1$ penalty on the vectors being convolved to separate the unknown sparse vectors. However, such approaches perform under additional restrictive assumptions on the convolved vectors. In \cite{lev11un}, a maximum-a-posteriori (MAP) based scheme is analyzed for image deblurring; the article illustrates the shortcomings of imposing sparsity enforcing priors on the gradients of natural images, and presents an alternative MAP estimator to recover only the blur kernel and then uses it to deblur the image.
In wireless communications, knowledge of the modulation scheme \cite{sato75me} or an estimate of the statistics of the source signal \cite{tong94bl} have been used for blind channel identification; these methods are overviewed in the review papers \cite{liu96re,tong98mu,johnson98bl,giannakis98ba}.  An effective scheme based on a deterministic model was put forth in \cite{xu95le}, where fundamental conditions for being able to identify multichannel responses from cross-correlations are presented.  
The work in this paper differs from this previous work in that it relies only on a single observation of two convolved signals, the model for these signals is that they lie in known (but arbitrary) subspaces rather than have a prescribed length, and we give a concrete relationship between the dimensions of these subspaces and the length of the observation sufficient for perfect recovery.

Recasting the quadratic problem in \eqref{eq:yconv} as the linear problem with a rank constraint in \eqref{eq:Al} is appealing since it puts the problem in a form for which we have recently acquired a tremendous amount of understanding.  Recovering a $N\times K$ rank-$R$ matrix from a set of linear observations has primarily been considered in two scenarios.  In the case where the observations come through a random projection, where either the $\mtx{A}_\ell$ are filled with independent Gaussian random variables or $\cA$ is an orthoprojection onto a randomly chosen subspace, the nuclear norm minimization program in \eqref{eq:nuclear-opt} is successful with high probability when\cite{recht10gu,candes11ti} 
\[
	L ~\geq~\mathrm{Const}\cdot R\max(K,N).
\]
When the observations are randomly chosen entries in the matrix, then subject to incoherence conditions on the singular vectors of the matrix being measured, the number of samples sufficient for recovery, again with high probability, is \cite{recht11si,candes10po,gross11re,candes09ex}
\[
	L ~\geq~ \mathrm{Const}\cdot R\max(K,N)\log^2(\max(K,N)).
\]
Our main result in Theorem~\ref{th:main} uses a completely different kind measurement system which exhibits a type of {\em structured randomness}; for example, when $\Bmat$ has the form \eqref{eq:Bembed}, $\cA$ has the concatenated Toeplitz structure shown in Figure~\ref{fig:multitoeplitz}.  In this paper, we will only be concerned with how well this type of operator can recover rank-1 matrices, ongoing work has shown that it also effectively recover general low-rank matrices \cite{ahmed12co}.

While this paper is only concerned with recovery by nuclear norm minimization, other types of recovery techniques have proven effective both in theory and in practice; see for example \cite{keshavan10ma,koltchinskii10nu,lee10ad}.  It is possible that the guarantees given in this paper could be extended to these other algorithms.

As we will see below, our mathematical analysis has mostly to do how matrices of the form in \eqref{eq:ymattime} act on rank-2 matrices in a certain subspace.  Matrices of this type have been considered in the context of sparse recovery in the compressed sensing literature for applications including multiple-input multiple-output channel estimation \cite{romberg10sp}, multi-user detection \cite{applebaum11as}, and multiplexing of spectrally sparse signals \cite{slavinsky11co}.
%
%
%
%
%
%
\section{Numerical Simulations}

In this section, we illustrate the effectiveness of the reconstruction algorithm for the blind deconvolution of vectors $\vct{x}$ and $\vct{w}$ with numerical experiments\footnote{MATLAB code that reproduces all of the experiments in this section is available at \url{http://www.aliahmed.org/code.html}.}. In particular, we study phase diagrams, which demonstrate the empirical probability of success over a range of dimensions $N$ and $K$ for a fixed $L$; an image deblurring experiment, where the task is to recover an image blurred by an unknown blur kernel; a channel protection experiment, where we show the robustness of our algorithm in the presence of additive noise. 

Some of the numerical experiments presented below are ``large scale'', with thousands (and even tens of thousands) of unknown variables.  Recent advances in SDP solvers, which we discuss in the following subsection, make the solution of such problems computationally feasible.

\subsection{Large-scale solvers}
\label{sec:largescale}

To solve the semidefinite program \eqref{eq:sdp-dual-dual} on instances where $K$ and $M$ are of practical size, we rely on the heuristic solver developed by Burer and Monteiro~\cite{Burer03}.  To implement this solver, we perform the variable substitution
\[
	\begin{bmatrix} \mtx{H} \\ \mtx{M} \end{bmatrix}	\begin{bmatrix} \mtx{H} \\ \mtx{M} \end{bmatrix}^*=
	\begin{bmatrix}
		\mtx{W}_1 &  \mtx{X}\\
		\mtx{X}^* & \mtx{W}_2
	\end{bmatrix}
\]
where $\mtx{H}$ is $K \times r$ and $\mtx{M}$ is $N \times r$ for $r>1$.  Under this substitution, the semidefinite constraint is always satisfied and we are left with the nonlinear program:
\begin{equation}
	\label{eq:burer-monteiro}
	\min_{\mtx{M},\mtx{H}}~\|\vct{M}\|_F^2 + \|\vct{H}\|_F^2 \quad 
	\mbox{subject to}\quad 
	\hat{\vct{y}} = \cA(\mtx{H}\mtx{M}^*),
	\quad \ell=1,\ldots,L.
\end{equation}
When $r=1$, this reformulated problem is equivalent to~\eqref{eq:least-squares-opt}.  Burer and Monteiro showed that provided $r$ is bigger than the rank of the optimal solution of~\eqref{eq:sdp-dual-dual}, all of the local minima of~\eqref{eq:burer-monteiro} are global minima of~\eqref{eq:sdp-dual-dual}~\cite{Burer05}.  Since we expect a rank one solution, we can work with $r=2$, declaring recovery when a rank deficient $\mtx{M}$ or $\mtx{H}$ is obtained.  Thus, by doubling the size of the decision variable, we can avoid the non-global local solutions of~\eqref{eq:least-squares-opt}.  Burer and Monteiro's algorithm has had notable success in matrix completion problems, enabling some of the fastest solvers for nuclear-norm-based matrix completion~\cite{Lee10,RechtRe11}.

To solve~\eqref{eq:burer-monteiro}, we implement the method of multipliers strategy initially suggested by Burer and Monteiro.  Indeed, this algorithm is explained in detail by Recht \emph{et al} in the context of solving problem~\eqref{eq:nuclear-opt}~\cite{recht10gu}.  The inner operation of minimizing the augmented Lagrangian term is performed using LBFGS as implemented by the Matlab solver minfunc~\cite{MINFUNC}.  This solver requires only being able to apply $\cA$ and $\cA^*$ quickly, both of which can be done in time $O(r \min\{N \log N,K\log K\})$.  The parameters of the augmented Lagrangian are updated according to the schedule proposed by Burer and Monteiro~\cite{Burer03}.  This code allows us to solve problems where $N$ and $K$ are in the tens of thousands in seconds on a laptop.
%
\subsection{Phase transitions}
\label{sec:phase}

Our first set of numerical experiments delineates the boundary, in terms of values for $K,N$ and $L$, for when \eqref{eq:nuclear-opt} is effective on generic instances of four different types of problems.  
For a fixed value of $L$, we vary the subspace dimensions $N$ and $K$ and run $100$ experiments, with different random instances of $\ww$ and $\xx$ for each experiment. The vectors $\hh$ and $\mm$ are selected to be standard Gaussian vectors with independent entries. Figures~\ref{fig:PhaseDiagramGaussian} and \ref{fig:PhaseDiagramSparse} show the collected frequencies of success for four different probabilistic models.  We classify a recovery a success if its relative error is less than 2\%\footnote{ The diagrams in Figures~\ref{fig:PhaseDiagramGaussian} and \ref{fig:PhaseDiagramSparse} do not change significantly if a smaller threshold, say on the order of $10^{-6}$, is chosen.}, meaning that if $\hat{\mtx{X}}$ is the solution to \eqref{eq:nuclear-opt}, then
\begin{equation}
	\label{eq:RMS}
	\frac{\|\hat{\mtx{X}}-\vct{w}\vct{x}^*\|_F}{\|\vct{w}\vct{x}^*\|_F} < 0.02.
\end{equation}

Our first set of experiments mimics the channel protection problem from Section~\ref{sec:channelprotect} and Figure~\ref{fig:channelprotect}.   Figure~\ref{fig:PhaseDiagramGaussian} shows the empirical rate of success when $\Cmat$ is taken as a dense $L\times N$ Gaussian random matrix.  We fix $L=2048$ and vary $N$ and $K$  from $25$ to $1000$.  
In Figure~\ref{fig:PhaseDiagramGaussian_sparse}, we take $\ww$ to be sparse with known support; we form $\Bmat$ by randomly selecting $K$ columns from the $L\times L$ identity matrix. For Figure~\ref{fig:PhaseDiagramGaussian_short}, we take $\ww$ to be ``short'', forming $\Bmat$ from the first $K$ columns of the identity.  In both cases, the basis expansion coefficient were drawn to be iid Gaussian random vectors.  In both cases, we are able to deconvolve this signals with a high rate of success when $L\gtrsim 2.7(K+N)$. 

\begin{figure}
	\centering
	\subfigure[]{
	\includegraphics[trim=2.5cm 7.5cm 2.5cm 7.5cm,scale = 0.4]{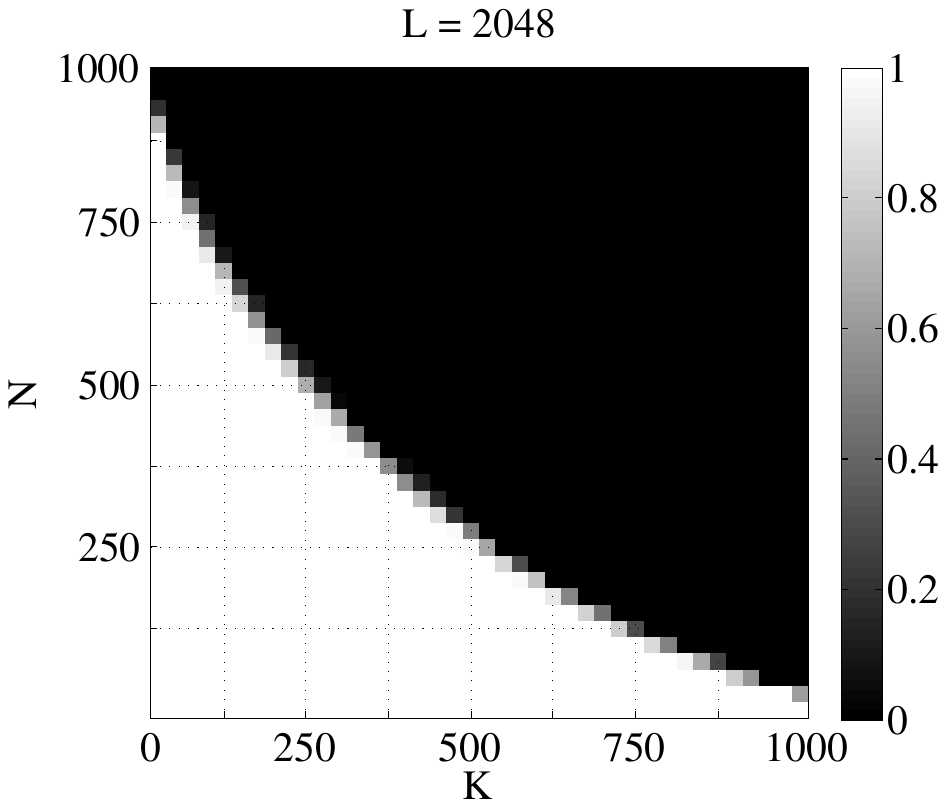}
	\label{fig:PhaseDiagramGaussian_sparse}}
	\subfigure[]{
  \includegraphics[trim=2.5cm 7.5cm 2.5cm 7.5cm,scale = 0.4]{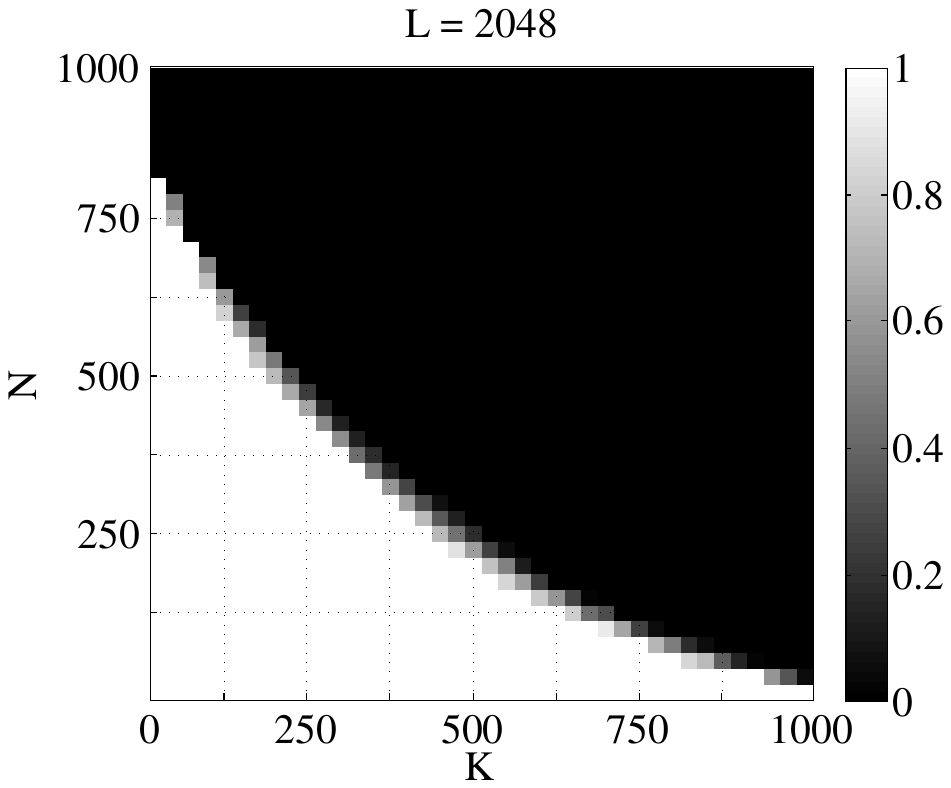}
  \label{fig:PhaseDiagramGaussian_short}}
  \caption{ Empirical success rate for the deconvolution of two vectors $\vct{x}$ and $\vct{w}$.  In these experiments, $\xx$ is a random vector in the subspace spanned by the columns of an $L\times N$ matrix whose entries are independent and identically distributed Gaussian random variables.  In part (a), $\ww$ is a generic sparse vector, with support and nonzero entries chosen randomly.  In part (b) $\ww$ is a generic short vector whose first $K$ terms are nonzero and chosen randomly.}
	\label{fig:PhaseDiagramGaussian}
\end{figure}
%
%
%
Figure~\ref{fig:PhaseDiagramSparse} shows the results of a similar experiment, only here both $\ww$ and $\xx$ are randomly generated sparse vectors.  We take $L$ to be much larger than the previous experiment, $L=32,768$, and vary $N$ and $K$ from $1000$ to $16,000$.  In Figure~\ref{fig:PhaseDiagramSparse_sparse}, we generate both $\Bmat$ and $\Cmat$ by randomly selecting columns of the identity --- despite the difference in the model for $\xx$ (sparse instead of randomly oriented) the resulting performance curve in this case is very similar to that in Figure~\ref{fig:PhaseDiagramGaussian_sparse}.  In Figure~\ref{fig:PhaseDiagramSparse_short}, we use the same model for $\Cmat$ and $\xx$, but use a ``short'' $\ww$ (first $K$ terms are non-zero).  Again, despite the difference in the model for $\xx$, the recovery curve looks almost identical to that in Figure~\ref{fig:PhaseDiagramGaussian_short}.

\begin{figure}[!ht]
	\centering
	\subfigure[]{
	      \includegraphics[trim=2.5cm 7.5cm 2.5cm 7.5cm,scale = 0.4]{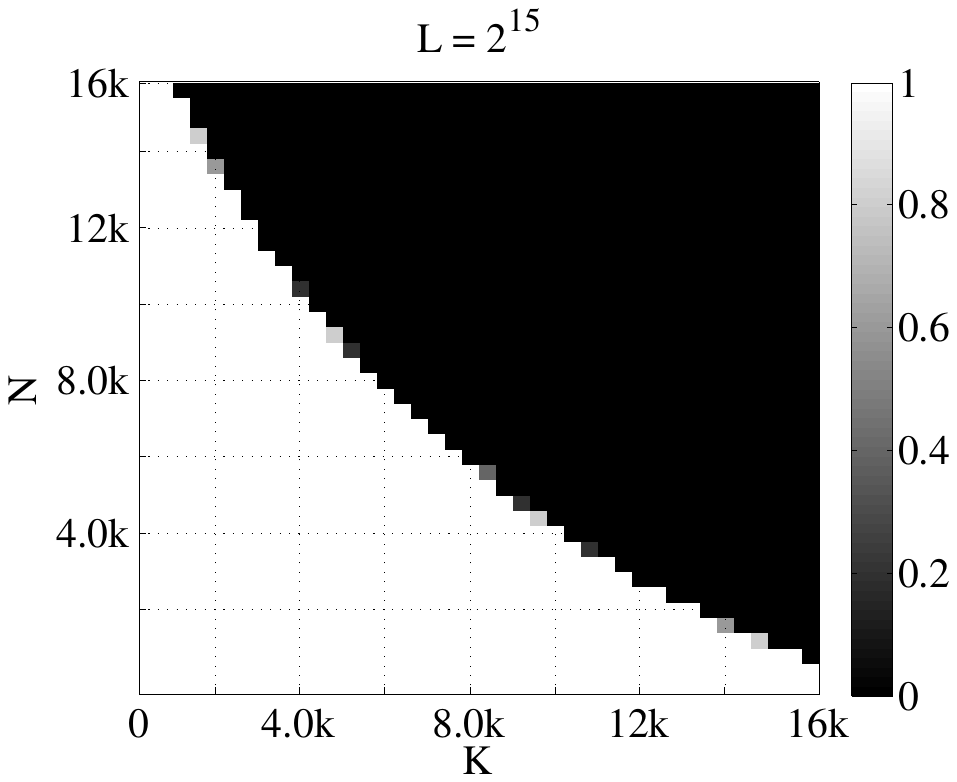}
	      \label{fig:PhaseDiagramSparse_sparse}}
	\subfigure[]{
        \includegraphics[trim=2.5cm 7.5cm 2.5cm 7.5cm,scale = 0.4]{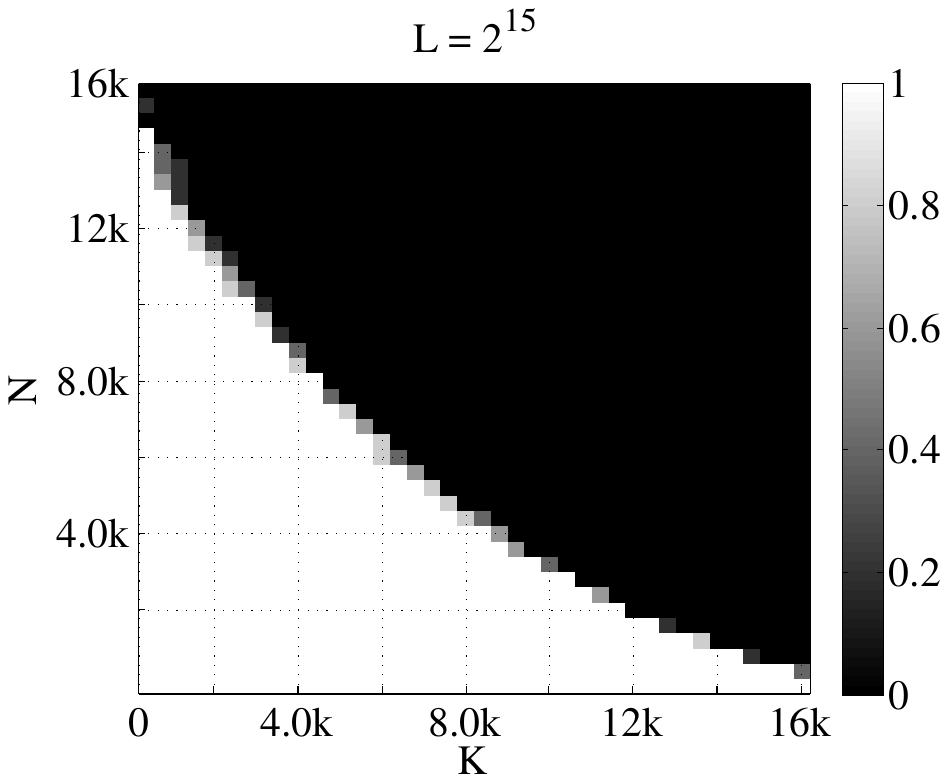}
        \label{fig:PhaseDiagramSparse_short}}
	\caption{ Empirical success rate for the deconvolution of two vectors $\vct{x}$ and $\vct{w}$.  In these experiments, $\xx$ is a random sparse vector whose support and $N$ non-zero values on that support are chosen at random.   In part (a), $\ww$ is a generic sparse vector, with support and $K$ nonzero entries chosen randomly.  In part (b) $\ww$ is a generic short vector whose first $K$ terms are nonzero and chosen randomly.
}
	\label{fig:PhaseDiagramSparse}
\end{figure}
%
\subsection{Recovery in the presence of noise}

Figure~\ref{fig:stable-recovery} demonstrates the robustness of the deconvolution algorithm in the presence of noise.  
We use the same basic experimental setup as in Figure~\ref{fig:PhaseDiagramGaussian_sparse}, with $L=2048$, $N=500$ and $K=250$, but instead of making a clean observation of $\ww\ast\xx$, we add a noise vector $\vct{z}$ whose entires are iid Gaussian with zero mean and variance $\sigma^2$.  We solve the program \eqref{eq:nnrelaxed} with $\delta = (L + \sqrt{4L})^{1/2}\sigma$, a value chosen since it will be an upper bound for $\|\vct{z}\|_2$ with high probability.  

Figure~\ref{fig:noisy1} shows how the relative error of the recovery changes with the noise level $\sigma$. On a log-log scale, the recovery error (show as $10\log_{10}\mbox{(relative error squared)}$) is linear in the signal-to-noise ratio (defined as SNR$=10\log_{10}(\|\vct{w}\vct{x}^*\|_F^2/\|\vct{z}\|_2^2$). For each SNR level, we calculate the average relative error squared over 100 iterations, each time using independent set of signals, coding matrix, and noise. Figure~\ref{fig:noisy2} shows how the recovery error is affected by the ``oversampling ratio''; as $L$ is made larger relative to $N+K$, the recovery error decreases. As before, each point is averaged over 100 independent iterations.

 \begin{figure}
 	\centering
 	\subfigure[]{
 	\includegraphics[trim=2.5cm 7.5cm 2.5cm 7.5cm, scale = 0.4]{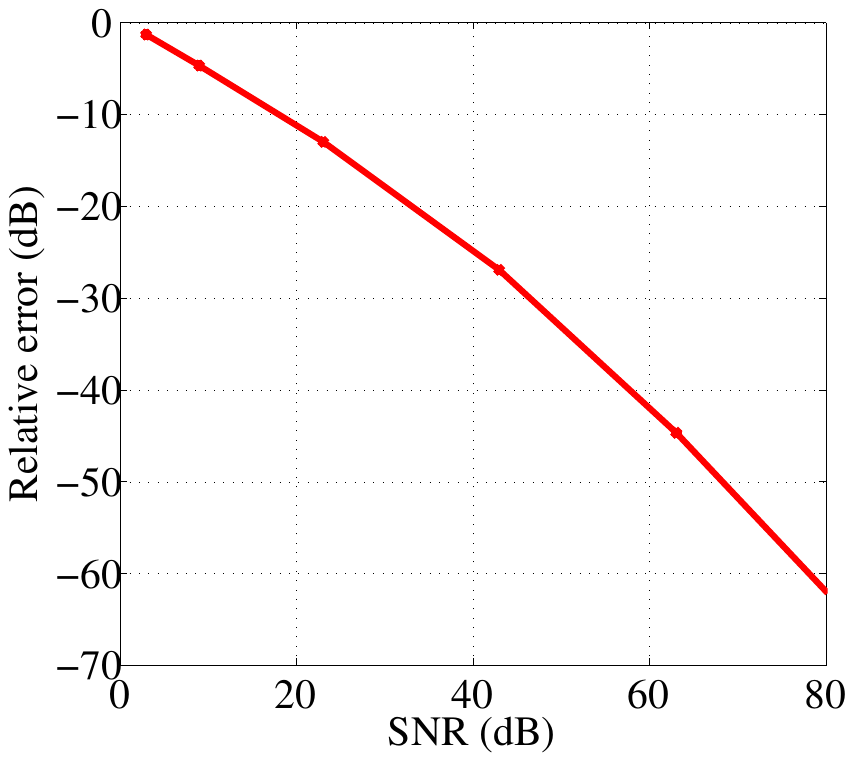}
 	\label{fig:noisy1}}
 	\subfigure[]{
 	\centering
 	\includegraphics[trim=2.5cm 7.5cm 2.5cm 7.5cm, scale = 0.4]{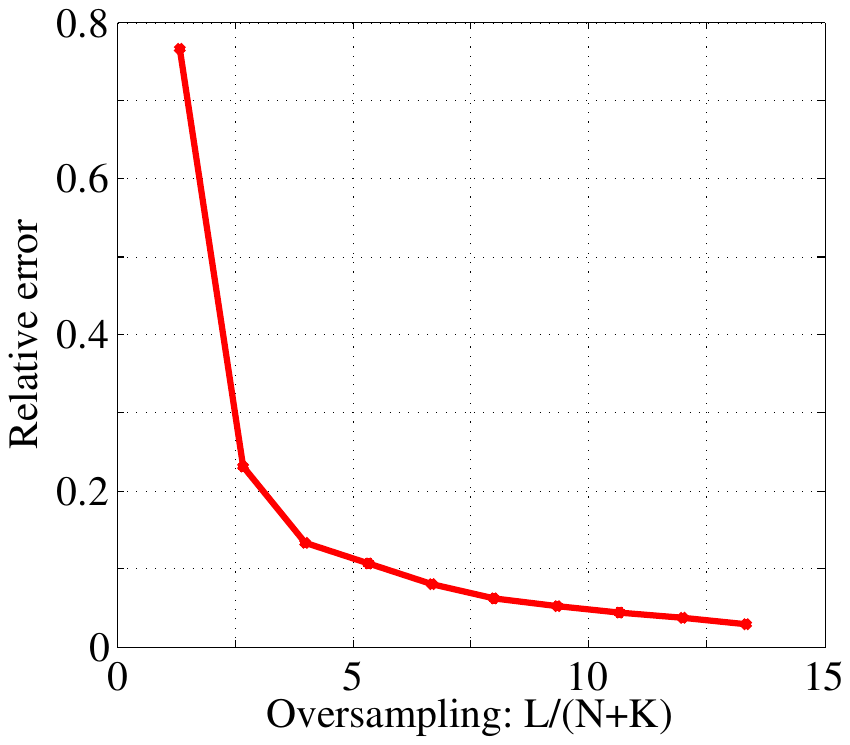}
 	\label{fig:noisy2}}
    \caption{ Performance of the blind deconvolution program in the presence of noise.  In all of the experiments, $L=2048$, $N=500$, $K=250$, $\Bmat$ is a random selection of columns from the identity, and $\Cmat$ is an iid Gaussian matrix.  (a) Relative error vs.\ SNR on a log-log scale. (b) Oversampling rate vs. relative error for a fixed SNR of $20dB$}
 	\label{fig:stable-recovery}
 \end{figure}
 
%
%
\subsection{Image deblurring}
\label{sec:imagedeblur}
The discrete signals $\ww$ and $\xx$ in the deconvolution problem \eqref{eq:yconv} may also represent higher-dimensional objects such as images. For example, the unknown $\xx \in \R^L$ may represent an image of the form $\xx[\ell_1,\ell_2]$, and the unknown $\ww \in \R^L$ may signify a 2D blur kernel $\ww[\ell_1,\ell_2]$, where $1 \leq \ell_1 \leq L_1, ~ 1 \leq \ell_2 \leq L_2$, and $L = L_1L_2$. The 2D convolution $\yy = \ww*\xx$ produces blurred image $\yy[\ell_1,\ell_2]$. Most natural images are sparse in some basis such as wavelets, DCT, or curvelets.  If we have an estimate of the active coefficients of the image $\xx$, then the image can be expressed as the multiplication of a small set of basis functions arranged as the columns of matrix $\mtx{C}$ and the corresponding short vector of active coefficients $\mm$, i.e., $\xx = \mtx{C}\mm$. In addition, if the non-zero components in the blur kernel $\ww$ are much smaller than the total number of pixels $L$, and we have an estimate of the support of the active components in $\ww$, then we can write $\ww = \mtx{B}\vct{h}$, where $\mtx{B}$ is the matrix formed by a subset of the columns of the identity matrix, and $\vct{h}$ is an unknown short vector. 

Figure~\ref{fig:Shapes}, \ref{fig:OracleShapesDeblur}, and \ref{fig:ActualShapesDeblur} illustrate an application of our blind deconvolution technique to two image deblurring problems.  In the first problem, we assume that we have oracle knowledge of a low-dimensional subspace in which the image to be recovered lies.  We observe a convolution of the $L = 65,536$ pixel Shapes image shown in Figure~\ref{fig:Shapes1} with the motion blurring kernel shown in  Figure~\ref{fig:Shapes2}; the observation is shown in Figure~\ref{fig:Shapes3}.  The Shapes image can be very closely approximated using only $N=5000$ terms in a Haar wavelet expansion, which capture 99.9\% of the energy in the image.  We start by assuming (perhaps unrealistically) that we know the indices for these most significant wavelet coefficients; the corresponding wavelet basis functions are taken as columns of $\Bmat$.  We will also assume that we know the support of the blurring kernel, which consists of $K=65$ connected pixels; the corresponding columns of the identity constitute $\Cmat$.  The image and blur kernel recovered by solving \eqref{eq:nuclear-opt} are shown in Figure~\ref{fig:OracleShapesDeblur}.

\begin{figure}[!ht]
	\centering	
	\subfigure[]{
	\includegraphics[trim=4.0cm 8.0cm 4.0cm 8.0cm,scale = 0.38]{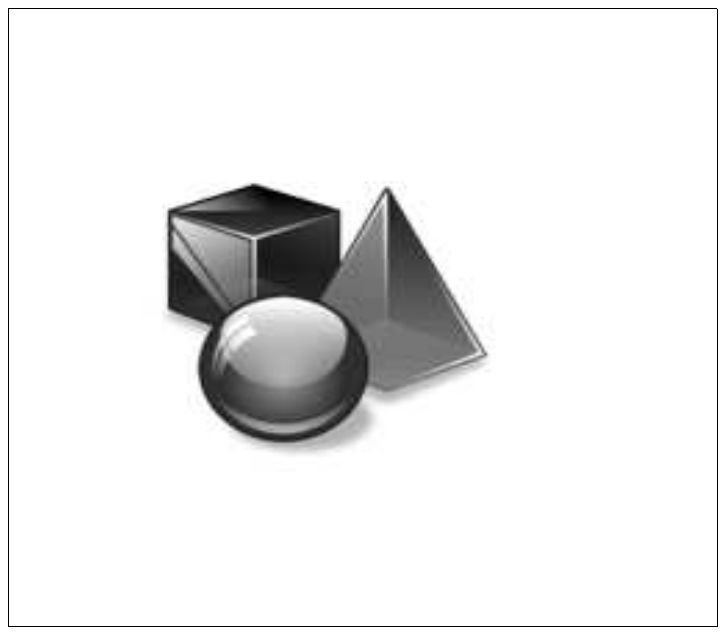}
	\label{fig:Shapes1}}
	\subfigure[]{
	\includegraphics[trim=4.0cm 8.0cm 4.0cm 8.0cm,scale = 0.38]{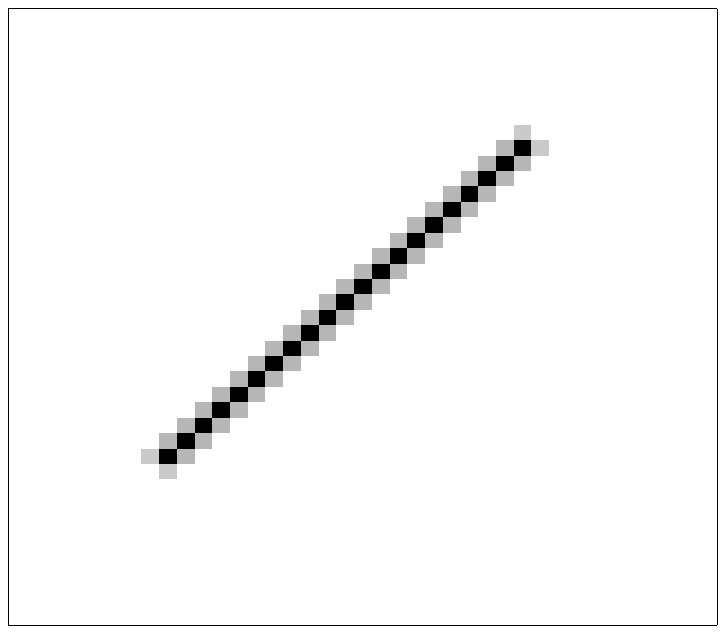}
	  \label{fig:Shapes2}}
	  \subfigure[]{
	  \includegraphics[trim=4.0cm 8.0cm 4.0cm 8.0cm,scale = 0.38]{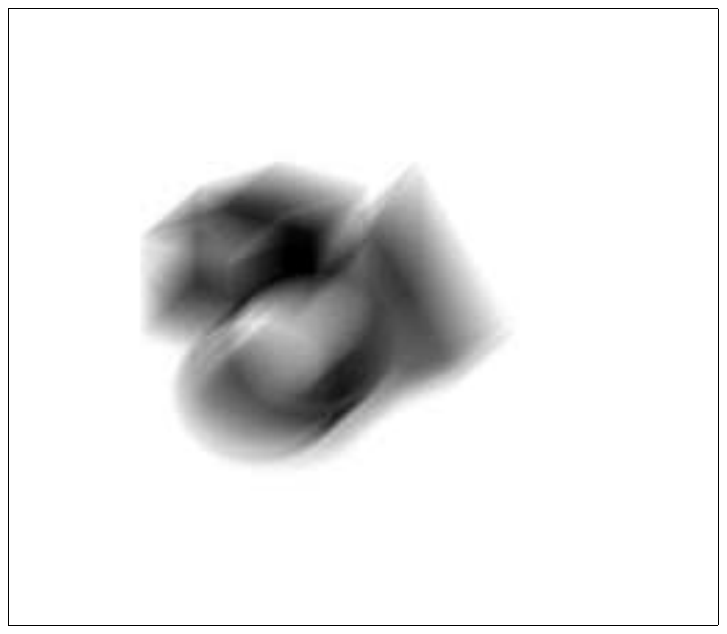}
	  \label{fig:Shapes3}}
	\caption{Shapes image for deblurring experiment. (a) Original $256 \times 256$ Shapes image $\xx$.  (b) Blurring kernel $\ww$ with a support size of 65 pixels, the locations of which are assumed to be known. (c) Convolution of (a) and (b).}
	\label{fig:Shapes}
\end{figure}

\begin{figure}[!ht]
	\centering
	\subfigure[]{
	\includegraphics[trim=4.0cm 8.0cm 4.0cm 8.0cm,scale = 0.45]{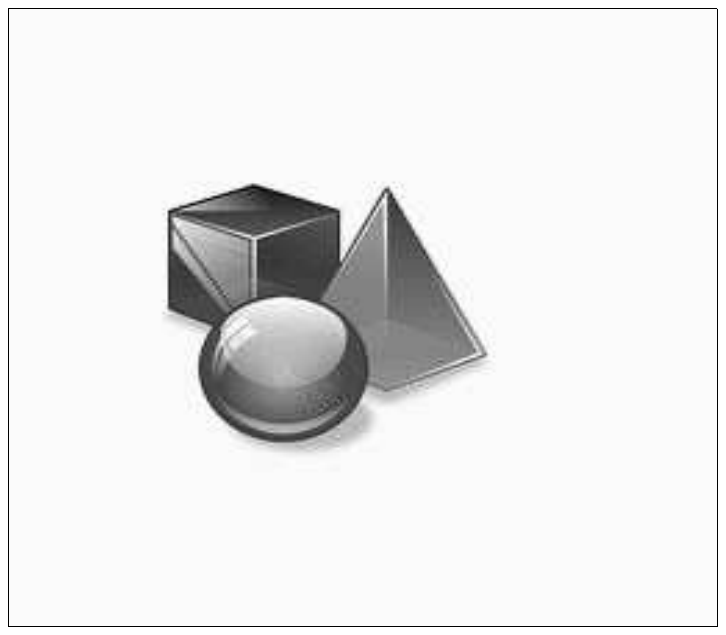}
	\label{fig:Shapes4}}
	\subfigure[]{
  \includegraphics[trim=4.0cm 8.0cm 4.0cm 8.0cm,scale = 0.45]{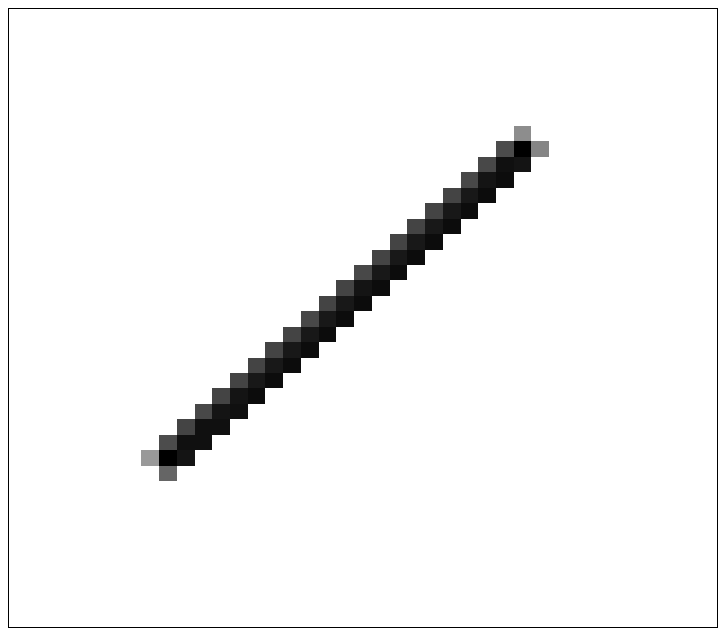}
  \label{fig:Shapes5}}\\
	\caption{An oracle assisted image deblurring experiment; we assume that we know the support of the $5000$ most significant wavelet coefficients of the original image. These wavelet coefficients capture 99.9\% of the energy in the original image. We obtain from the solution of \eqref{eq:nuclear-opt}: (a) Deconvolved image $\hat{\xx}$ obtained from the solution of \eqref{eq:nuclear-opt}, with relative error of $\|\hat{\xx}-\xx\|_2/\|\xx\|_2 = 1.6 \times 10^{-2}$. (b) Estimated blur kernel $\hat{\ww}$ with relative error of $\|\hat{\ww}-\ww\|_2/\|\ww\|_2 = 5.4 \times 10^{-1}$.}
	\label{fig:OracleShapesDeblur}
\end{figure}
\begin{figure}[!ht]
	\centering
		  \subfigure[]{
		  \includegraphics[trim=4.0cm 8.0cm 4.0cm 8.0cm,scale = 0.45]{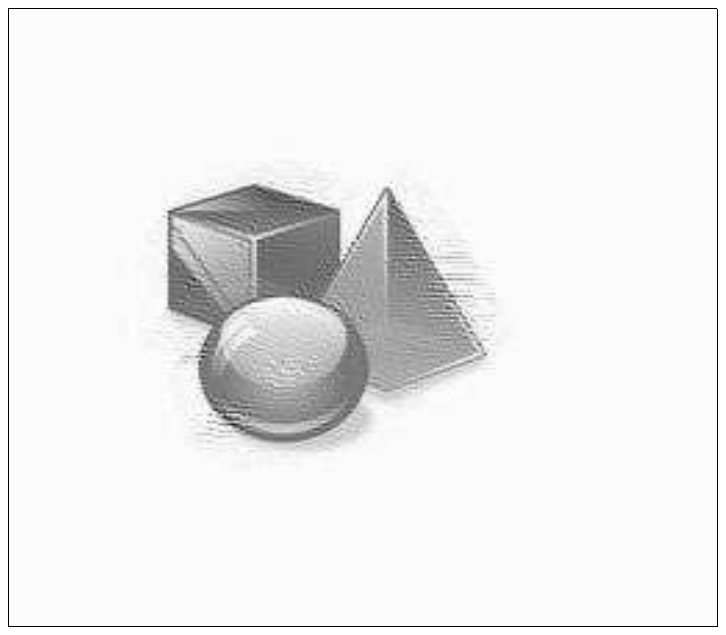}
		  \label{fig:Shapes6}}
		  \subfigure[]{
		  \includegraphics[trim=4.0cm 8.0cm 4.0cm 8.0cm,scale = 0.45]{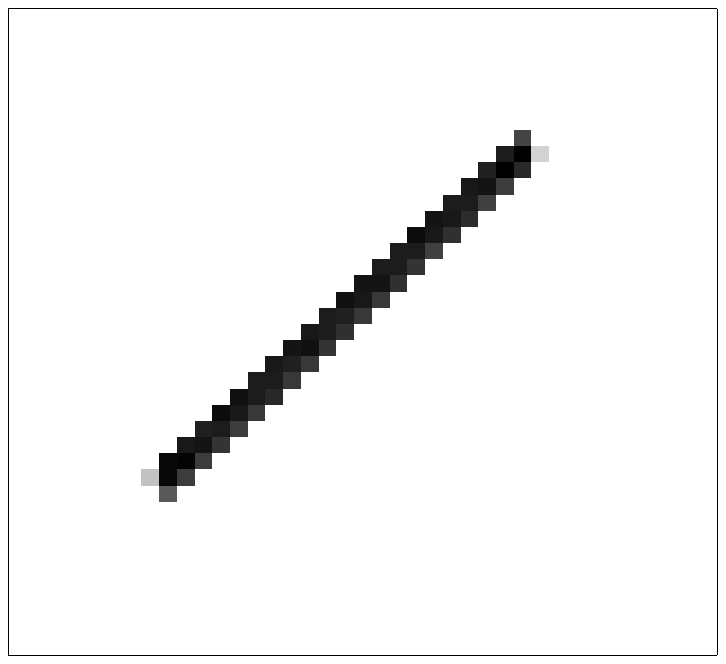}
		  \label{fig:Shapes7}}
		  \caption{Image recovery without oracle information. Take the support of the 9000 most-significant coefficients of Haar wavelet transform of the blurred image as our estimate of the subspace in which original image lives. (a) Deconvolved image obtained from the solution of \eqref{eq:nuclear-opt}, with relative error of $4.9 \times 10^{-2}$. (b) Estimated blur kernel; relative error = $5.6 \times 10^{-1}$.}
		  \label{fig:ActualShapesDeblur}
\end{figure}

Figure~\ref{fig:ActualShapesDeblur} shows a more realistic example where the support of the image in the wavelet domain is unknown.  We take the blurred image shown in Figure~\ref{fig:Shapes3} and, as before, we assume we know the support of the blurring kernel shown in Figure~\ref{fig:Shapes2}, with $K=65$ non-zero elements, but here we use the blurred image to estimate the support in the wavelet domain --- we take the Haar wavelet transform of the image in Figure~\ref{fig:Shapes3}, and select the indices of the $N=9000$ largest wavelet coefficients as a proxy for the support of the significant coefficients of the original image.  The wavelet coefficients of the original image at this estimated support capture 98.5\% of the energy in the blurred image.  The recovery using \eqref{eq:nuclear-opt} run with these linear models is shown in Figure~\ref{fig:Shapes6} and Figure~\ref{fig:Shapes7}. Despite not knowing the linear model explicitly, we are able to estimate it well enough from the observed data to get a reasonable reconstruction.

\section{Proof of main theorems}

In this section, we will prove Theorems~\ref{th:main} and \ref{th:stability} by establishing a set of standard sufficient conditions for $\mtx{X}_0$ to be the unique minimizer of \eqref{eq:nuclear-opt}.  At a high level, the argument follows previous literature \cite{candes09ex,gross10qu} on low-rank matrix recovery by constructing a valid {\em dual certificate} for the rank-1 matrix $\mtx{X}_0=\hh\mm^*$.  The main mathematical innovation in proving these results comes in Lemmas~\ref{lm:cAcAcond}, \ref{lm:PTAconditioning}, \ref{lm:pcoh} and \ref{lm:ApWpnorm}, which control the behavior of the random operator $\cA$.

Duality in convex programming is just one of the other known approaches employed to show the success of nuclear norm minimization in recovering the unknown matrix $\mtx{X}_0$. One method to establish the exact recovery is through the null-space properties of the linear operator $\cA$. In article \cite{recht11nu}, the authors characterize the null-space properties of $\cA$ that hold if and only if the solution to \ref{eq:nuclear-opt} equals $\mtx{X}_0$. In particular, the article demonstrates that a linear operator $\cA$ drawn from a Gaussian ensemble obeys the null-space properties with very high probability.  A closely related sufficient property to show the exact and stable recovery by solving the nuclear norm heuristic is the restricted isometry property (RIP) \cite{recht10gu} of $\cA$ over all rank-1 matrices. One way to prove such an RIP result for a linear operator $\cA$ is to show that $\cA$ acts as almost an isometry on a fixed matrix $\mtx{X}$ with very high probability, and then use a simple covering argument over all rank-1 matrices to establish the RIP. However, the construction of $\cA$ considered in this paper does not preserve the energy in a fixed matrix $\mtx{X}$ with appropriately high probability, hence, the argument does not work. Other methods to show an RIP result for the linear operator $\cA$ require further investigation; especially, for the case when $\ww$ is known to be sparse, and $\hat{\mtx{B}}$ is a partial Fourier matrix.

In this paper, we employ the dual certificate approach \cite{recht11si} to show the exact recovery of an unknown $\mtx{X}_0$ by solving the nuclear norm minimization program in \ref{eq:nuclear-opt}. We will work through the main argument in this section, leaving the technical details (including the proofs of the main lemmas) until Sections~\ref{sec:proofkeylemmas} and \ref{sec:supportinglemmas}.

Key to our argument is the subspace (of $\R^{K\times N}$) $T$ associated with $\mtx{X}_0=\hh\mm^*$:
\[
	T = \left\{\mtx{X} : \mtx{X} = \alpha \hh\vct{v}^* + 
	\beta \vct{u}\mm^*,~
	\vct{v}\in\R^{N},~\vct{u}\in\R^{K},~\alpha,\beta\in\R
	\right\}
\]	
with the (matrix) projection operators
\begin{align*}
	\PT(\mtx{X}) &= \mtx{P_H}\mtx{X} + \mtx{X}\mtx{P_M} - \mtx{P_H}\mtx{X}\mtx{P_M} \\
	\mathcal{P}_{T^\perp}(\mtx{X}) &= (\mtx{I} - \mtx{P_H})\mtx{X}(\mtx{I}-\mtx{P_M}), 
\end{align*}
where $\mtx{P_H}$ and $\mtx{P_M}$ are the (vector) projection matrices $\mtx{P_H} = \hh\hh^*$ and $\mtx{P_M} = \mm\mm^*$.

%
\subsection{Theorem~\ref{th:main}: Sufficient condition for a nuclear norm minimizer}

The following proposition is a specialization of the more general sufficient conditions for verifying the solutions to the nuclear norm minimization problem \eqref{eq:nuclear-opt} that have appeared multiple times in the literature in one form or another (see \cite{recht11si}, for example).
\begin{prop}
	\label{prop:dualdescent}
	The matrix $\mtx{X}_0=\hh\mm^*$ is the unique minimizer to \eqref{eq:nuclear-opt} if there exists a $\mtx{Y}\in\operatorname{Range}(\cA^*)$ such that
	\[
		\<\hh\mm^* - \PT(\mtx{Y}),\PT(\mtx{Z})\>_F - \<\PTc(\mtx{Y}),\PTc(\mtx{Z})\>_F + \|\PTc(\mtx{Z})\|_* 
		> 0
	\]
for all $\mtx{Z}\in\operatorname{Null}(\cA)$.
\end{prop}
For any two matrices $\mtx{A}$, $\mtx{B}$ with same dimensions, we will use the Holder's inequality:
\[
\<\mtx{A},\mtx{B}\>_F \leq \|\mtx{A}\|\|\mtx{B}\|_*,
\]
and the Cauchy-Schwartz's inequality:
\[
\<\mtx{A},\mtx{B}\>_F \leq \|\mtx{A}\|_F\|\mtx{B}\|_F.
\]
In view of the above inequalities, we have 
\begin{align*}
	\<\hh\mm^* - \PT(\mtx{Y}),& \PT(\mtx{Z})\>_F -  \<\PTc(\mtx{Y}),\PTc(\mtx{Z})\>_F + \|\PTc(\mtx{Z})\|_*
	\\
	&\geq -\|\hh\mm^*-\PT(\mtx{Y})\|_F\|\PT(\mtx{Z})\|_F - \|\PTc(\mtx{Y})\|\,\|\PTc(\mtx{Z})\|_* + \|\PTc(\mtx{Z})\|_*;
\end{align*}
therefore, it is enough to find a $\mtx{Y}\in\operatorname{Range}(\cA^*)$ such that
\begin{equation}
	\label{eq:Ysufficient1}
	-\|\hh\mm^*-\PT(\mtx{Y})\|_F\|\PT(\mtx{Z})\|_F + \left(1-\|\PTc(\mtx{Y})\|\right)\|\PTc(\mtx{Z})\|_*
	~>~ 0,
\end{equation}
for all $\mtx{Z}\in\operatorname{Null}(\cA)$.

In Lemma~\ref{lm:Anorm} in Section~\ref{sec:keylemmas} below we show that $\|\cA\|\leq\sqrt{(\alpha+1)N\log L}=:\gamma$ with probability at least $1-L^{-\alpha+1}$.  Corollary~\ref{cor:Tcond} below also shows that \eqref{eq:Mbound} implies 
\[
	\|\cA(\PT(\mtx{Z}))\|_F\geq 2^{-1/2}\|\PT(\mtx{Z})\|_F
	\quad\text{for all}\quad 
	\mtx{Z}\in\operatorname{Null}(\cA),  
\]
with high probability.  Then, since
\begin{align*}
	0 &= \|\cA(\mtx{Z})\|_F \\
	&\geq \|\cA(\PT(\mtx{Z}))\|_F - \|\cA(\PTc(\mtx{Z}))\|_F \\
	&\geq \frac{1}{\sqrt{2}}\,\|\PT(\mtx{Z})\|_F - \gamma\|\PTc(\mtx{Z})\|_F,
\end{align*}
we will have that
\begin{equation}
	\label{eq:PTZbound}
	\|\PT(\mtx{Z})\|_F ~\leq~ \sqrt{2}\gamma\|\PTc(\mtx{Z})\|_F 
	~\leq~ 
	\sqrt{2}\gamma\|\PTc(\mtx{Z})\|_*.
\end{equation}
Applying this fact to \eqref{eq:Ysufficient1}, we see that it is sufficient to find a $\mtx{Y}\in\operatorname{Range}(\cA^*)$ such that
\[
	\left(1 - \sqrt{2}\gamma\|\hh\mm^*-\PT(\mtx{Y})\|_F - \|\PTc(\mtx{Y})\|\right)\|\PTc(\mtx{Z})\|_* ~>~ 0.
\]
Since Lemma~\ref{lm:PTAconditioning} also implies that $\PTc(\mtx{Z})\not=\mtx{0}$ for $\mtx{Z}\in\operatorname{Null}(\cA)$, our approach will be to construct a $\mtx{Y}\in\operatorname{Range}(\cA^*)$ such that
\begin{equation}
	\label{eq:Yconds}
	\|\hh\mm^*-\PT(\mtx{Y})\|_F \leq \frac{1}{4\sqrt{2}\gamma}
	\quad\text{and}\quad
	 \|\PTc(\mtx{Y})\| < \frac{3}{4}.
\end{equation}

In the next section, we will show how such a $\mtx{Y}$ can be found using Gross's {\em golfing scheme} \cite{gross10qu,gross11re}.

\subsection{Construction of the dual certificate via golfing}
\label{sec:golfing}

The golfing scheme works by dividing the $L$ linear observations of $\mtx{X}_0$ into $P$ disjoint subsets of size $Q$, and then using these subsets of observations to iteratively construct the dual certificate $\mtx{Y}$.  We index these subsets by $\Gamma_1,\Gamma_2,\ldots,\Gamma_P$; by construction $|\Gamma_p|=Q$, $\bigcup_p\Gamma_p = \{1,\ldots,L\}$, and $\Gamma_p\cap\Gamma_{p'}=\emptyset$.  We define $\cA_p$ be the operator that returns the measurements indexed by the set $\Gamma_p$:
\[
	\cA_p(\mtx{W}) = \{\Tr{\ck\bk^*\mtx{W}}\}_{k\in\Gamma_p},\qquad
	\cA_p^*\cA_p\mtx{W} = \sum_{k\in\Gamma_p}\bk\bk^*\mtx{W}\ck\ck^*.
\]

The $\cA_p^*\cA_p$ are random linear operators; the expectation of their action on a fixed matrix $\mtx{W}$ is 
\[
	\E[\cA_p^*\cA_p\mtx{W}] = \sum_{k\in\Gamma_p}\bk\bk^*\mtx{W}. 
\]
\revise{
That the columns of $\Bmat$ are orthonormal and only the first $Q$ terms are nonzero gives us a natural way to choose the partition $\{\Gamma_p,~p=1,\ldots,P\}$.  Taking
\[
	\Gamma_p = \{(q-1)P+p,~q=1,\ldots,Q\},
\]
results in collections that correspond to sets of orthogonal rows in $\hat{\Bmat}$,
\begin{equation}
	\label{eq:Bportho}
	\sum_{k\in\Gamma_p}\bk\bk^* = \frac{Q}{L}\mtx{I},\quad\text{for all}~p=1,\ldots,P.
\end{equation}
Since each of the columns of $\Bmat$ is ``time limited'', the columns of $\hat\Bmat$ are ``bandlimited'', and so are isometrically preserved when every set of $Q$ equally spaced (modulo $L$) samples.
}

\revise{
Along with the expectation of each of the $\cA_p^*\cA_p$ being a multiple of the identity, we will also need tail bounds stating that $\cA_p^*\cA_p$ is close to its expectation with high probability.  These probabilities can be made smaller by making $Q$ larger.  As detailed below (in Lemmas~\ref{lm:PTAconditioning},\ref{lm:pcoh}, and \ref{lm:ApWpnorm}), taking
\begin{equation}
	\label{eq:Qbound1}
	Q ~=~ C_\alpha\,M\log(L)\log (M),
	\quad\text{where}\quad M = \max\left(\mu_{\max}^2K,\mu_h^2N\right),
\end{equation}
will make these probability bounds meaningful.  
}

The construction of $\mtx{Y}$ that obeys the conditions \eqref{eq:Yconds} relies on three technical lemmas which are stated below in Section~\ref{sec:keylemmas}.  Their proofs rely heavily on re-writing different quantities of interest (linear operators, vectors, and scalars) as a sum of independent subexponential random variables and then using a specialized version of the ``Matrix Bernstein Inequality'' to estimate their sizes.  Section~\ref{sec:concineq} below contains a brief overview of these types of probabilistic bounds.  The proofs of the key lemmas (\ref{lm:PTAconditioning}, \ref{lm:pcoh}, and \ref{lm:ApWpnorm}) are in Section~\ref{sec:proofkeylemmas}.  These proofs rely on several miscellaneous lemmas which compute simple expectations and tail bounds for various random variables; these are presented separately in Section~\ref{sec:supportinglemmas}.

With the $\Gamma_p$ chosen and the key lemmas established, we construct $\mtx{Y}$ as follows.  Let $\mtx{Y}_0 = \mtx{0}$, and then iteratively define
\[
	\mtx{Y}_p = \mtx{Y}_{p-1} + 
	\frac{L}{Q}\,\cA_p^*\cA_p\left(\vct{h}\vct{m}^* - \PT(\mtx{Y}_{p-1})\right).
\]
We will show that under appropriate conditions on $L$, taking $\mtx{Y}:=\mtx{Y}_P$ will satisfy both parts of \eqref{eq:Yconds} with high probability.

Let $\mtx{W}_p$ be the residual between $\mtx{Y}_p$ projected onto $T$ and the target $\vct{h}\vct{m}^*$:
\begin{equation*}
	\mtx{W}_p = \PT(\mtx{Y}_p) - \vct{h}\vct{m}^*.
\end{equation*}
Notice that $\mtx{W}_p\in T$ and
\begin{equation}
	\label{eq:Wpdef}
	\mtx{W}_0 = -\hh\mm^*,
	\qquad
	\mtx{W}_p = \frac{L}{Q}\left(\frac{Q}{L}\PT - \PT\cA_p^*\cA_p\PT\right)\mtx{W}_{p-1}.
\end{equation}
Applying Lemma~\ref{lm:PTAconditioning} iteratively to the $\mtx{W}_p$ tells us that
\begin{equation}
	\label{eq:Wpdecay}
	\|\mtx{W}_p\|_F ~\leq~ \frac{1}{2}\|\mtx{W}_{p-1}\|_F ~\leq~ 2^{-p}\|\vct{h}\vct{m}^*\|_F = 
	2^{-p},
	\qquad p=1,\ldots,P,
\end{equation}
with probability exceeding $1-3L^{-\alpha+1}$.  Thus we will have the first condition in \eqref{eq:Yconds},
\[
	\|\vct{h}\vct{m}^* - \PT(\mtx{Y}_P)\|_F \leq \frac{1}{4\sqrt{2}\gamma},
\]
for 
\[
	P ~=~ \frac{L}{Q} ~\geq~ \frac{\log(4\sqrt{2}\gamma)}{\log 2},
\]
which can be achieved with $Q$ as in \eqref{eq:Qbound1} and $M=\max(\mu_{\mathrm{max}}^2K,\mu_h^2N)$ as in \eqref{eq:Mbound}.

To bound $\|\PTc(\mtx{Y}_p)\|$, we use the expansion
\begin{align*}
	\mtx{Y}_p &= \mtx{Y}_{p-1} - \frac{L}{Q}\cA^*_p\cA_p\mtx{W}_{p-1} 
	~=~ \mtx{Y}_{p-2} - \frac{L}{Q}\cA^*_{p-1}\cA_{p-1}\mtx{W}_{p-2} - \frac{L}{Q}\cA^*_p\cA_p\mtx{W}_{p-1}
	~=~ \cdots \\
	&= -\sum_{p=1}^P\frac{L}{Q}\cA^*_{p}\cA_{p}\mtx{W}_{p-1} ,
\end{align*}
and so
\begin{align*}
	\|\PTc(\mtx{Y}_P)\| &= 
	\left\|\PTc\left(\sum_{p=1}^P\frac{L}{Q}\cA^*_p\cA_p\mtx{W}_{p-1}\right)\right\| \\
	&= \frac{L}{Q}\left\|\PTc\left(\sum_{p=1}^P\cA^*_p\cA_p\mtx{W}_{p-1} - \frac{Q}{L}\mtx{W}_{p-1}\right)\right\|,
	\qquad\text{(since $\mtx{W}_{p-1}\in T$)} \\
	&\leq \frac{L}{Q}\left\|\sum_{p=1}^P\cA^*_p\cA_p\mtx{W}_{p-1} - \frac{Q}{L}\mtx{W}_{p-1}\right\| \\
	&\leq \sum_{p=1}^P \frac{L}{Q}\left\|\cA^*_p\cA_p\mtx{W}_{p-1}-\frac{Q}{L}\mtx{W}_{p-1}\right\|.
\end{align*}
Lemma~\ref{lm:ApWpnorm} shows that with probability exceeding $1-L^{-\alpha+1}$,
\[
	\left\|\cA^*_p\cA_p\mtx{W}_{p-1}-\frac{Q}{L}\mtx{W}_{p-1}\right\| ~\leq~ 2^{-p}\frac{3Q}{4L},
	\quad\text{for all}~~~p=1,\ldots,P.
\]
and so
\[
	\|\PTc(\mtx{Y}_P)\| ~\leq~ \sum_{p=1}^P 3\cdot 2^{-p-2} ~<~ \frac{3}{4}.
\]

Collecting the results above, we see that both conditions in \eqref{eq:Yconds} will hold with probability exceeding $1-O(L^{-\alpha+1})$ when $M$ is chosen as in \eqref{eq:Mbound}.

\subsection{Theorem 2: Stability}

With the condition \eqref{eq:Mbound}, we know though the arguments in the previous section that with the required probability there will exist a dual certificate $\mtx{Y}$ that obeys the conditions \eqref{eq:Yconds} and that $\cA^*\cA$ is well conditioned on $T$: $\|\PT\cA^*\cA\PT-\PT\|\leq 1/2$.  


With these facts in place, the stability proof follows the template set in \cite{candes10ma,gross10qu}.  We start with two observations; first, the feasibility of $\mtx{X}_0$ implies
\begin{equation}
	\label{eq:cone}
	\|\tilde{\mtx{X}}\|_* \leq \|\mtx{X}_0\|_*,
\end{equation}
and
\begin{align}
	\label{eq:delta}
	\|\cA(\tilde{\mtx{X}}-\mtx{X}_0)\|_2 \leq
	\|\hat{\vct{y}}-\cA(\mtx{X}_0)\|_2+\|\cA(\tilde{\mtx{X}})-\hat{\vct{y}}\|_2
	&\leq 2\delta.
\end{align}
Set $\tilde{\mtx{X}} = \mtx{X}_0+\mtx{\xi}$. 
With $\mathcal{P}_{\cA}$ as the projection operator onto the row space of $\cA$, we break apart the recovery error as
\begin{align}
	\label{eq:xi-breakup}
	\|\mtx{\xi}\|_F^2 &= 	
	\|\mathcal{P}_{\cA}(\mtx{\xi})\|_F^2+\|\mathcal{P}_{\cA^\perp}(\mtx{\xi}) \|_F^2 \\
	&= \|\mathcal{P}_{\cA}(\mtx{\xi})\|_F^2 +
	\|\PT\mathcal{P}_{\cA^\perp}(\mtx{\xi})\|_F^2 + 
	\|\PTc\mathcal{P}_{\cA^\perp}(\mtx{\xi})\|_F^2\notag.
\end{align}

A direct result of of Proposition~\ref{prop:dualdescent} is that there exists a constant $C>0$ such that for all $\mtx{Z}\in\operatorname{Null}(\cA)$, $\|\mtx{X}_0+\mtx{Z}\|_*-\|\mtx{X}_0\|_* \geq C \|\PTc(\mtx{Z})\|_*$ (this is developed cleanly in \cite{recht11si}).  Since $\mathcal{P}_{\cA^\perp}(\mtx{\xi})\in\operatorname{Null}(\cA)$, we have 
\[
	\|\mtx{X}_0+\mathcal{P}_{\cA^\perp}(\mtx{\xi})\|_* - \|\mtx{X}_0\|_* \
	\geq C\|\PTc\mathcal{P}_{\cA^\perp}(\mtx{\xi})\|_*.
\]
Combining this with \eqref{eq:cone} and the triangle inequality yields
 \[
	\|\mtx{X}_0\|_* \geq \|\mtx{X}_0\|_* + 
	C\|\PTc\mathcal{P}_{\cA^\perp}(\mtx{\xi})\|_*- 
	\|\mathcal{P}_{\cA}(\mtx{\xi})\|_*,
\]
which implies
\begin{align*}
	\|\PTc\mathcal{P}_{\cA^\perp}(\mtx{\xi})\|_* 
	&\leq C\|\mathcal{P}_{\cA}(\mtx{\xi})\|_*\\ 
	&\leq C\sqrt{\min(K,N)}\|\mathcal{P}_{\cA}(\mtx{\xi})\|_F.
\end{align*}

In addition, in \eqref{eq:PTZbound} we established that for all $\mtx{Z}\in\operatorname{Null}(\cA)$, we have
\[
	\|\PT\mathcal{P}_{\cA^\perp}(\mtx{\xi})\|_F^2 
	~\leq~ 
	2\lambda_{\mathrm{max}}^2\|\PTc\mathcal{P}_{\cA^\perp}(\mtx{\xi})\|_F^2,
\]
and as a result
\[
	\|\mathcal{P}_{\cA^\perp}(\mtx{\xi})\|_F^2 
	~\leq~ 
	(2\lambda_{\mathrm{max}}^2+1)\|\PTc\mathcal{P}_{\cA^\perp}(\mtx{\xi})\|_F^2.
\]
Revisiting \eqref{eq:xi-breakup}, we have
\begin{align*}
	\|\tilde{\mtx{X}}-\mtx{X}_0\|_F^2 &\leq 
	(2\lambda_{\mathrm{max}}^2+1)\|\PTc\mathcal{P}_{\cA^\perp}(\mtx{\xi})\|_F^2 + 
	\|\mathcal{P}_{\cA}(\mtx{\xi})\|_F^2 \\
	&\leq C(2\lambda_{\mathrm{max}}^2+1)\min(K,N)
	\|\mathcal{P}_{\cA}(\mtx{\xi})\|_F^2+\|\mathcal{P}_{\cA}(\mtx{\xi})\|_F^2,
\end{align*}
and then absorbing all the constants into $C$,
\begin{align*}
	\|\tilde{\mtx{X}}-\mtx{X}_0\|_F &\leq 
	C\lambda_{\mathrm{max}}\sqrt{\min(K,N)}\|\mathcal{P}_{\cA}(\mtx{\xi})\|_F \\
	& \leq C\sqrt{\min(K,N)}\,\lambda_{\mathrm{max}}\|\cA^\dagger\|\,\|\cA(\mtx{\xi})\|_2,
\end{align*}
where $\cA^\dagger$ is the pseudo-inverse of $\cA$.
Using \eqref{eq:delta} and the fact that $\|\cA^\dagger\| = \lambda_{\mathrm{min}}^{-1}$, we obtain the final result
\begin{equation}
	\label{eq:hatX-error-bound}
	\|\tilde{\mtx{X}}-\mtx{X}_0\|_F \leq C \frac{\lambda_{\max}}{\lambda_{\min}}\sqrt{\min(K,N)}\delta.
\end{equation}
%
%
%
\subsection{Key lemmas}
\label{sec:keylemmas}

We start with two lemmas which characterize the singular values of the random linear operator $\cA$.  The first, which gives a loose upper bound on the maximum singular value, holds for all $N,K,L$.  The second gives a tighter bound on the maximum singular value and a comparable lower bound on the minimum singular value, but requires $\cA$ to be sufficiently underdetermined.

\begin{lem}[Operator norm of $\cA$]
	\label{lm:Anorm}
	Let $\cA$ be defined with $\Ak=\bk\ck^*$ as in Section~\ref{sec:mainresult}.  Fix $\alpha\geq 1$.
	Then 
	\[
		\|\cA\| ~\leq~ \sqrt{N(\log(NL/2) + \alpha\log L)},
	\]
	with probability exceeding $1-L^{-\alpha}$.
\end{lem}
\begin{proof}
	Writing $\cA$ in matrix form we have
	\begin{align*}
		\|\cA\|^2 ~=~ \left\|
		\begin{bmatrix}
			\Delta_1\hat{\Bmat} & \Delta_2\hat{\Bmat} & \cdots & \Delta_N\hat{\Bmat}
		\end{bmatrix}
		\right\|^2
		&=
		\left\|
		\begin{bmatrix}
			\hat{\Bmat}^*\Delta_1^* \\ \hat{\Bmat}^*\Delta_2^* \\ \vdots \\ \hat{\Bmat}^*\Delta_N^*
		\end{bmatrix}
		\right\|^2 \\
		&\leq 
		\left\|
		\begin{bmatrix}
			\Delta_1^* \\ \Delta_2^* \\ \vdots \\ \Delta_N^*
		\end{bmatrix}
		\right\| \\
		&\leq \|\Delta_1\|^2 + \|\Delta_2\|^2 + \cdots + \|\Delta_N\|^2 \\
		&\leq N\, \max_{1\leq n\leq N} \|\Delta_n\|^2 \\
		&= N\, \max_{1\leq n\leq N}\max_{1\leq\ell\leq L/2}|\hat{c}_\ell[n]|^2.
	\end{align*}
	Since the $|\hat{c}_\ell[n]|^2$ are independent chi-squared random variables,
	\[
		\P{\max_{n,\ell}|\hat{c}_\ell[n]|^2 > \lambda} ~\leq~ \frac{NL}{2}\,\er^{-\lambda},
	\]
	and the lemma follows by taking $\lambda = \log(NL/2) + \alpha\log L$.
\end{proof}
\begin{lem}[$\cA\cA^*$ is well conditioned.]
	\label{lm:cAcAcond}
	Let $\cA$ be as defined in \eqref{eq:Adef}, with coherences $\mu_{\max}^2$ and $\mu_{\min}^2$ as defined in \eqref{eq:mu1def} and \eqref{eq:muLdef}.  Suppose that $\cA$ is sufficiently underdetermined in that 
	\begin{equation}
		\label{eq:Lupper}
		NK ~\geq~ \frac{C_\alpha}{\mu^2_{\mathrm{min}}}L\log^2 L
	\end{equation}
	for some constant $C_\alpha>1$.  Then with probability exceeding $1-O(L^{-\alpha+1})$, the eigenvalues of $\cA\cA^*$ obey
	\[
		0.48\mu_{\min}^2\frac{NK}{L}\leq
		\lambda_{\mathrm{min}}(\cA\cA^*)\leq\lambda_{\mathrm{max}}(\cA\cA^*)\leq
		4.5\mu_{\max}^2\frac{NK}{L}.
	\]
\end{lem}

The proof of Lemma~\ref{lm:cAcAcond} in Section~\ref{sec:proofkeylemmas} decomposes $\cA\cA^*$ as a sum of independent random matrices, and then applies a Chernoff-like bound discussed in Section~\ref{sec:concineq}.


Our third key lemma tells us that, with high probability, the $\cA_p$ are well-conditioned when restricted to the subspace $T$.  The subsequent corollary shows that $\cA$ itself is also well-conditioned when restricted to $T$.

\begin{lem}[Conditioning on $T$]
	\label{lm:PTAconditioning}
	With the coherences $\mu_{\max}^2$ and $\mu_h^2$ defined in Section~\ref{sec:mainresult}, 
	let
	\begin{equation}
		\label{eq:Mdef}
		M = \max(\mu_{\max}^2K,\mu_h^2N).
	\end{equation}
	Fix $\alpha\geq 1$.  \revise{ Assume that the subsets $\Gamma_1,\ldots,\Gamma_P$ described in Section~\ref{sec:golfing} have size 
	\begin{equation}
		\label{eq:Qbound}
		|\Gamma_p| = Q ~=~ C'_\alpha\cdot M\log(L)\log(M),
	\end{equation}
	where $C'_\alpha = O(\alpha)$ is a constant chosen below.}  Then the linear operators $\cA_1,\ldots,\cA_P$ defined in Section~\ref{sec:golfing} will obey
	\[
		\max_{1\leq p\leq P} \left\|\PT\cA_p^*\cA_p\PT - \frac{Q}{L}\PT\right\| ~\leq~ \frac{Q}{2L},
	\]
	with probability exceeding $1-3PL^{-\alpha} \geq 1- 3L^{-\alpha+1}$.
\end{lem}

\begin{cor}
	\label{cor:Tcond}
	Let $\cA$ be the operator defined in \eqref{eq:Adef}, and $M$ be defined as in \eqref{eq:Mdef}.
	Then there exists a constant $C_\alpha = O(\alpha)$ such that 
	\begin{equation}
		\label{eq:Lcondlemma}
		M ~\leq~ \frac{L}{C_\alpha\log^2 L},
	\end{equation}
	implies
	\[
		\|\PT\cA^*\cA\PT - \PT\| ~\leq~ \frac{1}{2},
	\]
	with probability exceeding $1-3L^{-\alpha}$.
\end{cor}


\begin{lem}
	\label{lm:pcoh}
	Let $M$, $Q$, the $\Gamma_p$, and the $\cA_p$ be the same as in Lemma~\ref{lm:PTAconditioning}. 
	Let $\W_p$ be as in \eqref{eq:Wpdef}, and define
	\begin{equation}
		\label{eq:mupdef}
		\mu_p^2 = L\max_{\ell\in\Gamma_{p+1}}\|\mtx{W}_p^*\hat{\vct{b}}_\ell\|^2_2.
	\end{equation}
	Then there exists a constant $C_\alpha=O(\alpha)$ such that if
	\begin{equation}
		\label{eq:Lcoh}
		M ~\leq~ \frac{L}{C_\alpha \log^{3/2} L},
	\end{equation}
	then
	\begin{equation}
		\label{eq:mupdecay}
		\mu_p ~\leq~ \frac{\mu_{p-1}}{2},
		\quad\text{for}~~p=1,\ldots,P,
	\end{equation}
	with probability exceeding $1-2L^{-\alpha+1}$.
\end{lem}


\begin{lem}
	\label{lm:ApWpnorm}
	Let $\alpha$, $M$, $Q$, the $\Gamma_p$, and the $\cA_p$ be the same as in Lemma~\ref{lm:PTAconditioning}, and $\mu_p$ and $\mtx{W}_p$ be the same as in Lemma~\ref{lm:pcoh}.  Assume that \eqref{eq:Wpdecay} and \eqref{eq:mupdecay} hold:
	\[
		\|\W_{p-1}\|_F\leq 2^{-p+1}
		\quad\text{and}\quad
		\mu_{p-1}\leq 2^{-p+1}\mu_h.
	\]
	Then with probability exceeding $1-PL^{-\alpha}\geq 1 - L^{-\alpha+1}$,
	\revise{
	\[
		\left\|\cA_p^*\cA_p\mtx{W}_{p-1}-\frac{Q}{L}\mtx{W}_{p-1}\right\|
		~\leq~ 2^{-p}\frac{3Q}{4L},
		\quad\text{for all}~~p=1,\ldots,P.
	\]
	}
\end{lem}

\section{Concentration inequalities}
\label{sec:concineq}

Proving the key lemmas stated in Section~\ref{sec:keylemmas} revolves around estimating the sizes of sums of different subexponential random variables.  These random variables are either the absolute value of a sum of independent random scalars, the euclidean norm of a sum of independent random vectors (or equivalently, the Frobenius norm of a sum of random matrices), or the operator norm (maximum singular value) of a sum of random linear operators.  In this section, we very briefly overview the tools from probability theory that we will use to make these estimates.  The essential tool is the recently developed matrix Bernstein inequality \cite{tropp12us}.

We start by recalling the classical scalar Bernstein inequality.  A nice proof of the result in this form can be found in \cite[Chapter 2]{vandervaart96we}.
\begin{prop}[Scalar Bernstein, subexponential version]
	\label{prop:subexpbern}
	Let $z_1,\ldots,z_K$ be independent random variables with $\E[z_k]=0$, $\sigma_k^2: = \E[z_k^2]$, and 
	\begin{equation}
		\label{eq:scalarexp}
		\P{|z_k| > u} ~\leq~ C\er^{-u/\sigma_k},
	\end{equation}
	for some constants $C$ and $\sigma_k,~k=1,\ldots,K$ with
	\[
		\sigma^2 = \sum_{k=1}^K\sigma_k^2 \quad\text{and}\quad B = \max_{1\leq k\leq K} \sigma_k.
	\]
	Then
	\[
		\P{\left|z_1+\cdots+z_K\right| > u} ~\leq~
		2\exp\left(\frac{-u^2}{2C\sigma^2 + 2Bu}\right),
	\]
	and so
	\[
		|z_1 + \cdots + z_K| ~\leq~ 2\max\left\{\sqrt{C}\sigma\sqrt{t+\log 2},~2B(t+\log 2)\right\}
	\]
	with probability exceeding $1-\er^{-t}$.
\end{prop}

To make the statement (and usage) of the concentration inequalities more compact in the vector and matrix case, we will characterize subexponential vectors and matrices using their Orlicz-1 norm.
\begin{defn}\label{defn:mtxpsinorm}
	Let $\mtx{Z}$ be a random matrix.  We will use $\|\cdot\|_{\psi_1}$ to denote the Orlicz-1 norm:
	\[
		\|Z\|_{\psi_1} = \inf_{u\geq 0}\left\{\E[\exp(\|Z\|/u)]\leq 2\right\},
	\]
	where $\|\mtx{Z}\|$ is the spectral norm of $\mtx{Z}$.  In the case where $\mtx{Z}$ is a vector, we take $\|\mtx{Z}\|=\|\mtx{Z}\|_2$. 
\end{defn}

As the next basic result shows, the Orlicz-1 norm of a random variable can be systematically related to rate at which its distribution function approaches $1$ (i.e.\ $\sigma_k$ in \eqref{eq:scalarexp}).

\begin{lem}[Lemma~2.2.1 in  \cite{vandervaart96we}]
	\label{lm:tailtopsi1}
	Let $z$ be a random variable which obeys $\P{|z|>u}\leq \alpha\er^{-\beta u}$.  Then $\|z\|_{\psi_1}\leq (1+\alpha)/\beta$.
\end{lem}

Using these definitions, we have the following powerful tool for bounding the size of a sum of independent random vectors or matrices, each one of which is subexponential.  This result is mostly due to \cite{tropp12us}, but appears in the form below in \cite{koltchinskii10nu}.

\begin{prop}[Matrix Bernstein, Orlicz norm version]
	\label{prop:matbernpsi}
	Let $\mtx{Z}_1,\ldots,\mtx{Z}_Q$  be independent $K\times N$ random matrices with $\E[\mtx{Z_q}]=\mtx{0}$.  Let $B$ be an upper bound on the Orlicz-1 norms:
	\[
		\max_{1\leq q\leq Q}\|\mtx{Z}_q\|_{\psi_1} ~\leq~ B,
	\]
	and define
	\begin{equation}
		\label{eq:matbernsigma}
		\sigma^2 = \max\left\{\left\|\sum_{q=1}^Q\E[\mtx{Z}_q\mtx{Z}_q^*]\right\|,
		\left\|\sum_{q=1}^Q\E[\mtx{Z}_q^*\mtx{Z}_q]\right\|\right\}.
	\end{equation}
	Then there exists a constant $C$ such that for all $t\geq 0$
	\begin{equation}
		\label{eq:matbernpsi}
		\|\mtx{Z}_1+\cdots+\mtx{Z}_Q\| ~\leq~
		C\max\left\{~\sigma\sqrt{t + \log(K+N)},
		~B\log\left(\frac{\sqrt{Q}B}{\sigma}\right)
		(t+\log(K+N))~\right\},
	\end{equation}
	with probability at least $1-\er^{-t}$.
\end{prop}

Essential to establishing our stability result, Theorem~\ref{th:stability}, is bounding both the upper and lower eigenvalues of the operator $\cA\cA^*$.  We do this in Lemma~\ref{lm:cAcAcond} with a relatively straightforward application of the following Chernoff-like bound for sums of random positive symmetric matrices.

\begin{prop}[Matrix Chernoff in \cite{tropp12us}] 
	\label{prop:Chernoff} 
	Let $\mtx{Z}_1,\ldots,\mtx{Z}_Q$ be independent $L \times L$ random self-adjoint matrices whose eigenvalues obey
	\[
		0\leq\lambda_{\min}(\mtx{Z}_{q})\leq
		\lambda_{\max}(\mtx{Z}_{q})\leq R \quad \mbox{almost surely.}
	\]
Define 
\begin{equation*}
	\rho_{\min} : = \lambda_{\min}\left(\sum_{q = 1}^Q \E[\mtx{Z}_q]\right) 
	\quad \mbox{and} \quad 
	\rho_{\max} := \lambda_{\max}\left(\sum_{q = 1}^Q \E[\mtx{Z}_q]\right).
\end{equation*}
Then 
\begin{equation}
	\label{eq:Chernoff-min-bound}
	\P{\lambda_{\min}\left(\sum_{q=1}^Q \mtx{Z}_q\right)\leq t\rho_{\min}} 
	~\leq~ 
	L\,\er^{-(1-t)^2\rho_{\min}/2R}\quad \mbox{for} \quad t \in [0,1],
\end{equation}
and
\begin{equation}
	\label{eq:Chernoff-max-bound}
	\P{\lambda_{\max}\left(\sum_{q =1}^Q \mtx{Z}_q\right)\geq t\rho_{\max}} 
	~\leq~ 
	L\,\left[\frac{\er}{t}\right]^{t\rho_{\max}/R} \quad \mbox{for} \quad t \geq \er.
\end{equation}
\end{prop}

\section{Proof of key lemmas}
\label{sec:proofkeylemmas}

\subsection{Proof of Lemma~\ref{lm:cAcAcond}}

The proof of Lemma~\ref{lm:cAcAcond} is essentially an application of the matrix Chernoff bound in Proposition~\ref{prop:Chernoff}. 
	
Using the matrix form of $\cA$,
\[
	\cA = 
	\begin{bmatrix}
		\Delta_1\hat{\mtx{B}} & \Delta_2\hat{\mtx{B}} & \cdots & \Delta_N\hat{\mtx{B}}
	\end{bmatrix},
\]
we can write $\cA\cA^*$ as sum of random matrices
\[
	\cA\cA^* = \sum_{n = 1}^N \Delta_n \hat{\mtx{B}}\hat{\mtx{B}}^* \Delta_n^*, 
\]
where $\Delta_n = \operatorname{diag}(\{\hat{c}_{\ell}[n]\}_\ell)$ as in \eqref{eq:ymatfreq}.  To apply Proposition~\ref{prop:Chernoff}, we will need to condition on the maximum of the magnitudes of the $\hat{c}_\ell[n]$ not exceeding a certain size.  To this end, given an $\alpha$ (which we choose later), we define the event
\[
	\Gamma_\alpha = 
	\left\{\max_{\substack{ 1 \leq n \leq N \\ 1 \leq \ell \leq L/2}} 
	|\hat{c}_{\ell}[n]| \leq \alpha \right\},
\]
and since the $|\hat{c}_{\ell}|^2$ are Rayleigh random variables,
\[
	\P{\Gamma_\alpha^c} ~\leq~ \frac{NL}{2}e^{-\alpha^2}.
\]
We can now breakdown the calculation as
\begin{align}
	\label{eq:Ytmax}
	\P{\lambda_{\max}(\cA\cA^*) > v} &\leq 
	\P{\lambda_{\max}(\cA\cA^*)  > v \mid \Gamma_\alpha}\,\P{\Gamma_\alpha} +
	\P{\Gamma_\alpha^c} \\
	&\leq \P{\lambda_{\max}(\cA\cA^*)  > v \mid \Gamma_\alpha} + \P{\Gamma_\alpha^c},
\end{align}
and similarly for $\P{\lambda_{\min}(\cA\cA^*)>v}$.  Conditioned on $\Gamma_\alpha$, the complex Gaussian random variables $\hat{c}_\ell[n]$ are still zero mean and independent; we denote these conditional random variables as $\hat{c}'_\ell[n]$, and set $\Delta'_n = \operatorname{diag}(\{\hat{c}'_{\ell}[n]\}_\ell)$, noting that
\[
	\E[|\hat{c}'_\ell[n]|^2] = \E[|\hat{c}_\ell[n]|^2\mid\Gamma_\alpha] = \frac{1-(\alpha^2+1)\er^{-\alpha^2}}{1-\er^{-\alpha^2}}=:\sigma_\alpha^2 \leq 1.
\]

We now apply Proposition~\ref{prop:Chernoff} with 
\begin{align*}
	R = \max_n\left\{\lambda_{\max}(\Delta'_n \hat{\mtx{B}}\hat{\mtx{B}}^* \Delta_n'^*)\right\} 
	&\leq\max_n\left\{ \lambda_{\max}(\Delta'_n)
	\lambda_{\max}(\hat{\mtx{B}}\hat{\mtx{B}}^*)\lambda_{\max}(\Delta_n'^*)\right\} 
	\leq \alpha^2,
\end{align*}
and
\begin{align*}
	\rho_{\max} &= \lambda_{\max}\left(\sum_{n = 1}^N
	\E[\Delta_n'\hat{\Bmat}\hat{\Bmat}^*\Delta_n'^*]\right)
 	= N\lambda_{\max}\left(\E[\Delta_n'\hat{\Bmat}\hat{\Bmat}^*\Delta_n'^*]\right) 
	~\leq~ N\sigma_\alpha^2\max_{\ell}\|\hat{\vct{b}}_\ell\|^2_2 
	~=~ \mu_{\max}^2N\frac{K}{L},
\end{align*}
and
\begin{align*}
	\rho_{\min} &= \lambda_{\min}\left(\sum_{n = 1}^N 
	\E\left[\Delta'_n \hat{\Bmat}\hat{\Bmat}^* \Delta_n'^*\right]\right)
 	~=~ N\sigma_\alpha^2\min_\ell \|\hat{\vct{b}}_\ell\|_2^2
	~=~ \sigma_\alpha^2\mu_{\min}^2N\frac{K}{L},
\end{align*}
which yields
\[
	\P{\left.\lambda_{\min}(\cA\cA^*)  < \frac{\sigma_\alpha^2\mu_{\min}^2NK}{2L} ~\right\vert~ \Gamma_\alpha} 
	~\leq~ L\exp\left(-\frac{\sigma_\alpha^2\mu_{\min}^2NK}{8\alpha^2L}\right),
\]
where we have take $t=1/2$ in \eqref{eq:Chernoff-min-bound}, and
\[
	\P{\left.\lambda_{\max}(\cA\cA^*) > \frac{\er^{3/2}\mu_{\max}^2NK}{L} ~\right\vert~ \Gamma_\alpha} 
	~\leq~
	L\exp\left(-\frac{2\mu_{\max}^2NK}{\alpha^2L}\right),
\]
where we have taken $t=\er^{3/2}$ in \eqref{eq:Chernoff-max-bound}.  Then taking $\alpha = \sqrt{2\log L}$  establishes the lemma.


\subsection{Proof of Lemma~\ref{lm:PTAconditioning}}
\label{sec:PTcondproof}

The proof of the Lemma and its corollary follow the exact same line of argumentation.  We will start with the conditioning of the partial operators $\cA_p$ on $T$; after this, the argument for the conditioning of the full operator $\cA$ will be clear.

We start by fixing $p$, and set $\Gamma=\Gamma_p$.  With
\[
	\Ak = \bk\ck^*,
\]
where the $\bk\in\C^K$ obey \eqref{eq:Biso},\eqref{eq:mu1def},\eqref{eq:muhdef} and the $\ck\in\C^N$ are random vectors distributed as in \eqref{eq:normalck}, we are interested in how the random operator
\[
	\PT\cA_p^*\cA_p\PT = \sum_{k\in\Gamma}\PT(\Ak)\otimes\PT(\Ak)
\]
concentrates around its mean in the operator norm.  This operator is a sum of independent random rank-1 operators on $N\times K$ matrices, and so we can use the matrix Bernstein inequality in Proposition~\ref{prop:matbernpsi} to estimate its deviation.

Since $\Ak = \bk\ck^*$, $\PT(\Ak)$ is the rank-2 matrix given by 
\begin{align*}
	\PT(\Ak) &= \<\bk,\hh\>\hh\ck^* + \<\mm,\ck\>\bk\mm^* - \<\bk,\hh\>\<\mm,\ck\>\hh\mm^* \\
	&= \hh\vk^* + \uk\mm^*,
\end{align*}
where
$\vk = \<\hh,\bk\>\ck$ and $\uk=\<\mm,\ck\>(\bk-\<\bk,\hh\>\hh) = \<\mm,\ck\>(\I-\hh\hh^*)\bk$.

The linear operator $\PT(\cdot)$, since it maps $K\times N$ matrices to $K\times N$ matrix, can itself be represented as a $KN\times KN$ matrix that operates on a matrix that has been rasterized (in column order here) into a vector of length $KN$.  We will find it convenient to denote these matrices in block form: $\{M(i,j)\}_{i,j}$, where $M(i,j)$ is a $K\times K$ matrix that occupies rows $(i-1)K+1,\ldots,iK$ and columns $(j-1)K+1,\ldots,jK$.  Using this notation, we can write $\PT$ as the matrix 
\begin{equation}
	\label{eq:PTmatrix}
	\PT = \{\hh\hh^*\delta(i,j)\}_{i,j} + \{m[i]m[j]\I\}_{i,j} -
	\{m[i]m[j]\hh\hh^*\}_{i,j},
\end{equation}
where $\delta(i,j) = 1$ if $i=j$ and is zero otherwise.

We will make repeated use the following three facts about block matrices below:
\begin{enumerate}
	
	\item Let $\mathcal{M}$ be an operator that we can write in matrix form as
	\[
		\mathcal{M} = \{\mtx{M}\delta(i,j)\}_{i,j}
	\]
	for some $K\times K$ matrix $\mtx{M}$.  Then the action of $\mathcal{M}$ on a matrix $\mtx{X}$ is
	\[
		\mathcal{M}(\mtx{X}) = \mtx{M}\mtx{X},
	\]
	and so $\|\mathcal{M}\| = \|\mtx{M}\|$.  Also, $\mathcal{M}^*(\mtx{X}) = \mtx{M}^*\mtx{X}$.
	
	\item Now suppose we can write $\mathcal{M}$ in matrix form as
	\[
		\mathcal{M} = \{p[i]^*q[j]\I\}_{i,j},
	\]
	for some $\vct{p},\vct{q}\in\C^N$.  Then the action of $\mathcal{M}$ on a matrix $\mtx{X}$ is
	\[
		\mathcal{M}(\mtx{X}) = \mtx{X}\vct{q}\vct{p}^*,
	\]
	and so $\|\mathcal{M}\| = \|\vct{q}\vct{p}^*\| = \|\vct{q}\|_2\|\vct{p}\|_2$.  Also, $\mathcal{M}^*(\mtx{X}) = \mtx{X}\vct{p}\vct{q}^*$.
	
	\item Now let 
	\[
		\mathcal{M} = \{p[i]^*q[j]\mtx{M}\}_{i,j}.
	\]
	Then the action of $\mathcal{M}$ on a matrix $\mtx{X}$ is
	\[
		\mathcal{M}(\mtx{X}) = \mtx{M}\mtx{X}\vct{q}\vct{p}^*,
	\]
	and so $\|\mathcal{M}\| = \|\mtx{M}\|\,\|\vct{q}\vct{p}^*\| = \|\mtx{M}\|\,\|\vct{q}\|_2\|\vct{p}\|_2$.  Also $\mathcal{M}^*(\mtx{X}) = \mtx{M}^*\mtx{X}\vct{p}\vct{q}^*$.
\end{enumerate}

We will break $\PT(\Ak)\otimes\PT(\Ak)$ into four different tensor products of rank-1 matrices, and treat each one in turn:
\begin{equation}
	\label{eq:PTPTexpanded}
	\PT(\Ak)\otimes\PT(\Ak) = \hh\vk^*\otimes\hh\vk^* + \hh\vk^*\otimes\uk\mm^* + 
	\uk\mm^*\otimes\hh\vk^* + \uk\mm^*\otimes\uk\mm^*.
\end{equation}
To handle these terms in matrix form, note that if $\uu_1\vv_1^*$ and $\uu_2\vv_2^*$ are rank-1 matrices, with $\uu_i\in\C^K$ and $\vv_i\in\C^N$, then the operator given by their tensor product can be written as
\[
	\uu_1\vv_1^*\otimes\uu_2\vv_2^* = 
	\begin{bmatrix}
		v_1[1]^*v_2[1]\uu_1\uu_2^* & v_1[1]^*v_2[2]\uu_1\uu_2^* & \cdots & v_1[1]^*v_2[N]\uu_1\uu_2^* \\
		v_1[2]^*v_2[1]\uu_1\uu_2^* & v_1[2]^*v_2[2]\uu_1\uu_2^* & \cdots & v_1[2]^*v_2[N]\uu_1\uu_2^* \\
		\vdots & & \ddots & \\
		v_1[N]^*v_2[1]\uu_1\uu_2^* & \cdots & \cdots & v_1[N]^*v_2[N]\uu_1\uu_2^*
	\end{bmatrix}
	= \left\{v_1[i]^*v_2[j]\uu_1\uu_2^*\right\}_{i,j}.
\]

For the expectation of the sum, we compute the following:
\begin{align*}
	\E[\hh\vk^*\otimes\hh\vk^*] &= |\<\hh,\bk\>|^2\,\E[\{\hat{c}_k[i]^*\hat{c}_k[j]\vct{h}\vct{h}^*\}_{i,j}] \\
	&= |\<\hh,\bk\>|^2\,\{\delta(i,j)\hh\hh^*\}_{i,j},
\end{align*}
and
\begin{align*}
	\E[\uk\mm^*\otimes\uk\mm^*] &= \E[|\<\mm,\ck\>|^2]\,
	\{m[i]m[j](\I-\hh\hh^*)\bk\bk^*(\I-\hh\hh^*)\}_{i,j} \\
	&= \{m[i]m[j](\I-\hh\hh^*)\bk\bk^*(\I-\hh\hh^*)\}_{i,j},
\end{align*}
since $\E[|\<\mm,\ck\>|^2] = \|\mm\|^2_2 = 1$, and
\begin{align*}
	\E[\hh\vk^*\otimes\uk\mm^*] &= \E\{v_k[i]^*m[j]\hh\uk^*\}_{i,j} \\
	&= \<\bk,\hh\>\,\{\E[\hat{c}_k[i]^*\<\ck,\mm\>] m[j]\hh\bk^*(\I-\hh\hh^*)\}_{i,j} \\
	&= \<\bk,\hh\>\,\{m[i]m[j]\hh\bk^*(\I-\hh\hh^*)\}_{i,j},
\end{align*}
and
\begin{align*}
	\E[\uk\mm^*\otimes\hh\vk^*] &= 
	\<\hh,\bk\>\,\{\E[\hat{c}_k[j]\<\mm,\ck\>]m[i](\I-\hh\hh^*)\bk\hh^*\}_{i,j} \\
	&= \<\hh,\bk\>\,\{m[i]m[j](\I-\hh\hh^*)\bk\hh^*\}_{i,j}.
\end{align*}
A straightforward calculation combines these four results with \eqref{eq:PTmatrix} to verify that
\[
	\E[\PT(\Ak)\otimes\PT(\Ak)] = \PT(\{\bk\bk^*\delta(i,j)\}_{i,j}\PT).
\]
\revise{
In light of \eqref{eq:Bportho}, this means
\begin{align}
	\label{eq:PTAPTAmean}
	\E\left[\PT\cA_p^*\cA_p\PT\right] &=
	\E\left[\sum_{k\in\Gamma}\PT(\Ak)\otimes\PT(\Ak)\right] = 
	\frac{Q}{L}\PT.
\end{align}
}

We now derive tail bounds for how far the sum over $\Gamma$ for each of the terms in \eqref{eq:PTPTexpanded} deviates from their respective means.  Starting with first term, we use the compact notation
\[
	\cZ_k = \hh\vk^*\otimes\hh\vk^* - \E[\hh\vk^*\otimes\hh\vk^*],
\]
for each addend.  To apply Proposition~\ref{prop:matbernpsi}, we need to uniformly bound the size (Orlicz $\psi_1$ norm) of each individual $\cZ_k$ as well as the variance $\sigma^2$ in \eqref{eq:matbernsigma}.  For the uniform size bound,
\begin{align*}
	\|\cZ_k\| &= |\<\hh,\bk\>|^2\left\|\{(\hat{c}_k[i]^*\hat{c}_k[j]-\delta(i,j))\hh\hh^*\}_{i,j}\right\| \\
	&= |\<\hh,\bk\>|^2\left\|\{(\hat{c}_k[i]^*\hat{c}_k[j]-\delta(i,j))\I\}\{\hh\hh^*\delta(i,j)\}_{i,j}\right\| \\
	&\leq |\<\hh,\bk\>|^2\,\|\hh\hh^*\|\,\|\ck\ck^*-\I\| \\
	&\leq \frac{\mu_h^2}{L}\max(\|\ck\|^2_2,1).
\end{align*}
Applying Lemma~\ref{lm:maxck1tail},
\[
	\P{\max(\|\ck\|^2_2,1) > u} ~\leq~ 1.2\,\er^{-u/8N},
\]
and combined with Lemma~\ref{lm:tailtopsi1} this means
\[
	\|\cZ_k\|_{\psi_1} ~\leq~ \frac{\mu_h^2}{L}\|\max(\|\ck\|^2_2,1)\|_{\psi_1} 
	~\leq~ C\,\frac{\mu_h^2N}{L}.
\]
For the variance, we need to compute $\E[\cZ_k^*\cZ_k]$.  This will be easiest if we rewrite the action of $\cZ_k$ on a matrix $\mtx{X}$ as
\[
	\cZ_k(\mtx{X}) = |\<\hh,\bk\>|^4\hh\hh^*\mtx{X}(\ck\ck^*-\I),
\]
and so
\[
	\cZ_k^*\cZ_k(\mtx{X}) = |\<\hh,\bk\>|^4\|\hh\|^2_2\hh\hh^*\mtx{X}(\ck\ck^*-\I)^2,
\]
and
\begin{align*}
	\E[\cZ_k^*\cZ_k(\mtx{X})] &= |\<\hh,\bk\>|^4\|\hh\|^2_2\hh\hh^*\mtx{X}\E[(\ck\ck^*-\I)^2] \\
	&= N|\<\hh,\bk\>|^4\hh\hh^*\mtx{X},
\end{align*}
and finally
\revise{
\begin{align*}
	\left\|\sum_{k\in\Gamma}\E[\cZ_k^*\cZ_k]\right\| &=
	N\sum_{k\in\Gamma}|\<\hh,\bk\>|^4 \\
	&\leq \frac{\mu_h^2N}{L}\sum_{k\in\Gamma}|\<\hh,\bk\>|^2 \\
	&=\frac{\mu_h^2NQ}{L^2},
\end{align*}
where we have used \eqref{eq:Bportho} in the last step.  
}
Collecting these results and applying Proposition~\ref{prop:matbernpsi} with $t=\alpha\log L$ yields
\begin{align}
	&\left\|\sum_{k\in\Gamma}\hh\vk^*\otimes\hh\vk^* - \E[\hh\vk^*\otimes\hh\vk^*]\right\|
	\leq\notag\\
	&C_\alpha\, \frac{\mu_h\sqrt{N\log L}}{L}\max\left\{\sqrt{Q},\mu_h\sqrt{N\log L}\log(\mu_h^2N)\right\}\label{eq:hvkterm},
\end{align}
with probability exceeding $1-L^{-\alpha}$.

For the sum over the second term in \eqref{eq:PTPTexpanded}, set
\begin{align*}
	\cZ_k &= \uk\mm^*\otimes\uk\mm^* - \E[\uk\mm^*\otimes\uk\mm^*] \\
	&= \left(|\<\mm,\ck\>|^2-1\right)\{m[i]m[j](\I-\hh\hh^*)\bk\bk^*(\I-\hh\hh^*)\}_{i,j},
\end{align*}
then using the fact that $\|\I-\hh\hh^*\|\leq 1$ (since $\|\hh\|_2=1$), we have
\begin{align*}
	\|\cZ_k\| &= \left|\,|\<\mm,\ck\>|^2 - 1\right|\, \|(\I-\hh\hh^*)\bk\|_2^2\, \|\mm\|^2_2 \\
	&\leq \left|\,|\<\mm,\ck\>|^2 - 1\right|\,\|\bk\|^2_2 \\
	&\leq \left|\,|\<\mm,\ck\>|^2 - 1\right|\,\frac{\mu_{\max}^2 K}{L}.
\end{align*}
This is again a subexponential random variable whose size we can characterize using Lemma~\ref{lm:ckv2subexp}:
\[
	\||\<\mm,\ck\>|^2 - 1\|_{\psi_1}\leq C
	\quad\text{and so}\quad
	\|\cZ_k\|_{\psi_1} ~\leq~
	C\,\frac{\mu_{\max}^2 K}{L}.
\]
To bound the variance in \eqref{eq:matbernpsi}, we again write out the action of $\cZ_k$ on an arbitrary $K\times N$ matrix $\mtx{X}$:
\[
	\cZ_k(\mtx{X}) = (|\<\mm,\ck\>|^2-1)(\I-\hh\hh^*)\bk\bk^*(\I-\hh\hh^*)\mtx{X}\mm\mm^*,
\]
and so
\begin{align*}
	\E[\cZ_k^*\cZ_k(\mtx{X})] &= 
	\E[(|\<\mm,\ck\>|^2-1)^2]\|(\I-\hh\hh^*)\bk\|^2_2(\I-\hh\hh^*)\bk\bk^*(\I-\hh\hh^*)\mtx{X}\mm\mm^* \\
	&= \|(\I-\hh\hh^*)\bk\|^2_2(\I-\hh\hh^*)\bk\bk^*(\I-\hh\hh^*)\mtx{X}\mm\mm^*,
\end{align*}
where in the last step we have used the fact that $|\<\mm,\ck\>|^2$ is a chi-square random variable with two degrees of freedom with variance $\E[(|\<\mm,\ck\>|^2-1)^2]=1$.  This gives us
\revise{
\begin{align*}
	\left\|\sum_{k\in\Gamma}\E[\cZ_k^*\cZ_k]\right\| &=
	\left\|\sum_{k\in\Gamma}\|(\I-\hh\hh^*)\bk\|^2_2(\I-\hh\hh^*)\bk\bk^*(\I-\hh\hh^*)
	\right\| \\
	&\leq \max_{k\in\Gamma}\left(\|(\I-\hh\hh^*)\bk\|^2_2\right)
	\left\|\sum_{k\in\Gamma}(\I-\hh\hh^*)\bk\bk^*(\I-\hh\hh^*) \right\| \\
	&\leq \frac{\mu_{\max}^2 K}{L}\left\|\sum_{k\in\Gamma}\bk\bk^*\right\| \\
	&= \frac{\mu_{\max}^2 KQ}{L^2}.
\end{align*}
}
Collecting these results and applying Proposition~\ref{prop:matbernpsi} with $t=\alpha\log L$ yields
\begin{align}
	&\left\|\sum_{k\in\Gamma}\uk\mm^*\otimes\uk\mm^* - \E[\uk\mm^*\otimes\uk\mm^*]\right\| \leq\notag\\	&C_\alpha\frac{\mu_{\max}\sqrt{K\log L}}{L}\max\left\{\sqrt{Q},\mu_{\max}\sqrt{K\log L}\log(\mu_{\max}^2K)\right\}	\label{eq:ukmterm},
\end{align}
with probability exceeding $1-L^{-\alpha}$.

The last two terms in \eqref{eq:PTPTexpanded} are adjoints of one another, so they will have the same operator norm.  We now set
\begin{align*}
	\cZ_k &= \hh\vk^*\otimes\uk\mm^* - \E[\hh\vk^*\otimes\uk\mm^*] \\
	&= \<\hh,\bk\>\{m[i](\hat{c}_k[j]\<\mm,\ck\> - m[j])(\I-\hh\hh^*)\bk\hh^*\}_{i,j},
\end{align*}
and so the action of $\cZ_k$ on an arbitrary matrix $\mtx{X}$ is given by
\[
	\cZ_k(\mtx{X}) = \<\hh,\bk\>(\I-\hh\hh^*)\bk\hh^*\mtx{X}(\ck\ck^*-\I)\mm\mm^*,
\]
from which we can see
\begin{align*}
	\|\cZ_k\| &\leq |\<\hh,\bk\>|\,\|\bk\|_2\|(\ck\ck^*-\I)\mm\|_2 \\
	&\leq \frac{\mu_h\mu_{\max}\sqrt{K}}{L}\,\|(\ck\ck^*-\I)\mm\|_2.
\end{align*}
From Lemmas~\ref{lm:ckmcktail} and \ref{lm:tailtopsi1}, we that the random variable $\|(\ck\ck^*-\I)\mm\|_2$ is subexponential with $\|(\ck\ck^*-\I)\mm\|_{\psi_1}\leq C\sqrt{N}$, and so
\[
	\|\cZ_k\|_{\psi_1} ~\leq~ C\frac{\mu_h\mu_{\max}\sqrt{KN}}{L}.
\]
For the variance $\sigma^2$ in \eqref{eq:matbernsigma}, we need to bound the sizes of both $\cZ_k^*\cZ_k$ and $\cZ_k\cZ_k^*$.  Starting with the former, we have 
\begin{align*}
	\E[\cZ_k^*\cZ_k(\mtx{X})] &=
	|\<\hh,\bk\>|^2\|(\I-\hh\hh^*)\|_2^2\hh\hh^*\mtx{X}\E[(\ck\ck^*-\I)\mm\mm^*(\ck\ck^*-\I)],
\end{align*}
and then applying Lemma~\ref{lm:Eckckv} yields
\revise{
\begin{align*}
	\left\|\sum_{k\in\Gamma}\E[\cZ_k^*\cZ_k]\right\| &=
	\sum_{k\in\Gamma} |\<\hh,\bk\>|^2\,\|(\I-\hh\hh^*)\bk\|^2_2 \\
	&\leq \sum_{k\in\Gamma} |\<\hh,\bk\>|^2\,\|\bk\|^2_2 \\
	&\leq \frac{\mu_{\max}^2K}{L}\sum_{k\in\Gamma} |\<\hh,\bk\>|^2 \\
	&= \frac{\mu_{\max}^2KQ}{L^2}.
\end{align*}
}
For $\cZ_k\cZ_k^*$,
\begin{align*}
	\E[\cZ_k\cZ_k^*(\mtx{X})] &=
	|\<\hh,\bk\>|^2(\I-\hh\hh^*)\bk\bk^*(\I-\hh\hh^*)\mtx{X}\mm\mm^*\E[(\ck\ck^*-\I)^2]\mm\mm^*,
\end{align*}
and then applying Lemma~\ref{lm:EckckI2} yields
\revise{
\begin{align*}
	\left\|\sum_{k\in\Gamma}\E[\cZ_k\cZ_k^*]\right\| &=
	N\left\|(\I-\hh\hh^*)\left(\sum_{k\in\Gamma}|\<\hh,\bk\>|^2\bk\bk^*\right)(\I-\hh\hh^*)\right\| \\
	&\leq N\left\|\sum_{k\in\Gamma}|\<\hh,\bk\>|^2\bk\bk^*\right\| \\
	&\leq \frac{\mu_h^2N}{L}\left\|\sum_{k\in\Gamma}\bk\bk^*\right\| \\
	&= \frac{\mu_h^2NQ}{L^2}.
\end{align*}
}
Collecting these results and applying Proposition~\ref{prop:matbernpsi} with $t=\alpha\log L$ and 
$M = \max\left\{\mu_{\max}^2K, \mu_h^2 N\right\}$ yields
\begin{align}
	\label{eq:crossterm}
	\left\|\sum_{k\in\Gamma}\hh\vk^*\otimes\uk\mm^* - \E[\hh\vk^*\otimes\uk\mm^*]\right\|
	&\leq
	C_\alpha\,\frac{\sqrt{M\log L}}{L}
	\max\left\{\sqrt{Q}, \sqrt{M\log L}\log(M)\right\},
\end{align}
with probability exceeding $1-L^{-\alpha}$.

\revise{
We can combine \eqref{eq:PTAPTAmean} with \eqref{eq:hvkterm}, \eqref{eq:ukmterm}, and \eqref{eq:crossterm} to establish that
\begin{align*}
	\left\|\PT\cA_p^*\cA_p\PT - \frac{Q}{L}\PT\right\| &
	\leq
	C_\alpha\,\frac{\sqrt{M\log L}}{L}
	\max\left\{\sqrt{Q}, \sqrt{M\log L}\log(M)\right\},
\end{align*}
with probability exceeding $1-3L^{-\alpha}$.  
}
With $Q$ chosen as in \eqref{eq:Qbound}, this becomes
\revise{
\begin{align*}
	\left\|\PT\cA_p^*\cA_p\PT - \frac{Q}{L}\PT\right\| &\leq
	C_\alpha\,\frac{Q}{L}\,\max\left\{\frac{1}{\sqrt{C'_\alpha\log M}},\frac{1}{C'_\alpha}\right\}  \\
	&\leq \frac{Q}{2L},
\end{align*}
}
for $C'_\alpha$ chosen appropriately.  Applying the union bound establishes the lemma.

\revise{
To prove the corollary, we take $\Gamma=\{1,\ldots,L\}$ and $Q=L$ above.}  We have
\begin{align*}
	\left\|\PT\cA^*\cA\PT - \PT\right\| &\leq
	C_\alpha\,\max\left\{\sqrt{\frac{M\log L}{L}},\frac{M\log(L)\log(M)}{L}\right\},
\end{align*}
with probability exceeding $1-3L^{-\alpha}$.  Then taking $L$ as in \eqref{eq:Lcondlemma} will guarantee the desired conditioning.

\subsection{Proof of Lemma~\ref{lm:pcoh}}

We start by fixing $\ell\in\Gamma_{p+1}$ and estimating $\|\W_p^*\hat{\vct{b}}_\ell\|_2$.  We can re-write $\W_p$ as a sum of independent random matrices: since $\W_{p-1}\in T$, $\PT(\W_{p-1}) = \W_{p-1}$ and
\revise{
\begin{align*}
	\W_p &= \PT\left(\cA_p^*\cA_p\W_{p-1} - \frac{Q}{L}\W_{p-1}\right) \\
	&= \PT\left(\sum_{k\in\Gamma_p}\bk\bk^*\W_{p-1}\ck\ck^* - \sum_{k\in\Gamma_p}\bk\bk^*\W_{p-1}\right)  \\
	&=\sum_{k\in\Gamma_p}\PT(\Zk) ,
\end{align*}
where $\Zk = \bk\bk^*\W_{p-1}(\ck\ck^*-\I)$.  Thus 
\begin{align}
	\label{eq:Wpblnorm}
	\|\W_p^*\hat{\vct{b}}_\ell\|_2 &\leq \left\|\sum_{k\in\Gamma_p}\hat{\vct{b}}_\ell^*\PT(\Zk)\right\|_2,
\end{align}
with the right-hand side as the norm of a sum of independent zero-mean random vectors which we will bound using Propositions~\ref{prop:subexpbern} and  \ref{prop:matbernpsi}. 
}
We set $\vct{w}_k = \W_{p-1}^*\bk$ and expand $\hat{\vct{b}}_\ell^*\PT(\Zk)$ as
\begin{align*}
	\hat{\vct{b}}_\ell^*\PT(\Zk) = \<\hh,\hat{\vct{b}}_\ell\>&\<\bk,\hh\>\vct{w}_k^*(\ck\ck^*-\I) + \\
	&+ \<\bk,\hat{\vct{b}}_\ell\>\vct{w}_k^*(\ck\ck^*-\I)\mm\mm^* -
\<\hh,\hat{\vct{b}}_\ell\>\<\bk,\hh\>\vct{w}_k^*(\ck\ck^*-\I)\mm\mm^*,
\end{align*}
and so
\begin{align}
	\label{eq:vectorscalar}
	\left\|\sum_{k\in\Gamma_p}\hat{\vct{b}}_\ell^*\PT(\Zk)\right\|_2 &\leq
	\left\|\sum_{k\in\Gamma_p}\vct{z}_k\right\|_2 + 
	\left|\sum_{k\in\Gamma_p}z_k\right|,
\end{align}
where the $\vct{z}_k$ are independent random vectors, and the $z_k$ are independent random scalars:
\[
	\vct{z}_k = \<\hat{\vct{b}}_\ell,\hh\>\<\hh,\bk\>(\ck\ck^*-\I)\vct{w}_k,
	\qquad
	z_k = \<\bk,(\I-\hh\hh^*)\hat{\vct{b}}_\ell\>\,\<(\ck\ck^*-\I)\mm,\vct{w}_k\>.
\]
Using Lemma~\ref{lm:ckvuck}, we have a tail bound for each term in the scalar sum:
\begin{align*}
	\P{|z_k| > \lambda} &\leq
	2\er\cdot\exp\left(-\frac{\lambda}{\|\vct{w}_k\|_2|\<\bk,(\I-\hh\hh^*)\hat{\vct{b}}_\ell\>|}\right).
\end{align*}
Applying the scalar Bernstein inequality (Proposition~\ref{prop:subexpbern}) with
\begin{align*}
	B &= \max_{k} \|\vct{w}_k\|_2|\<\bk,(\I-\hh\hh^*)\hat{\vct{b}}_\ell\>| ~\leq~
	\frac{\mu_{p-1}\mu_{\max}^2K}{L^{3/2}},
\end{align*}
and
\revise{
\begin{align*}
	\sigma^2 &= \sum_{k\in\Gamma_p}\|\vct{w}_k\|^2_2|\<\bk,(\I-\hh\hh^*)\hat{\vct{b}}_\ell\>|^2 \\
	&\leq \frac{\mu_{p-1}^2}{L}\sum_{k\in\Gamma_p}|\<\bk,(\I-\hh\hh^*)\hat{\vct{b}}_\ell\>|^2 \\
	&= \frac{\mu_{p-1}^2Q}{L^2}\|(\I-\hh\hh^*)\hat{\vct{b}}_\ell\|^2_2 \\
	&\leq \frac{\mu_{p-1}^2\mu_{\max}^2KQ}{L^3},
\end{align*}
and taking } $t = \alpha\log L$
tells us that
\begin{align}
	\label{eq:mupscalarbound}
	\left|\sum_{k\in\Gamma_p}z_k\right| &\leq 
	C_\alpha\,\frac{\mu_{p-1}\mu_{\max}\sqrt{K\log L}}{L^{3/2}}\max\left\{\sqrt{Q}, \mu_{\max}\sqrt{K\log L}\right\},
\end{align}
with probability at least $1-L^{-\alpha}$.

For the vector term in \eqref{eq:vectorscalar}, we apply Lemmas~\ref{lm:ckmcktail} and \ref{lm:tailtopsi1} to see that
\begin{align*}
	\|\vct{z}_k\|_{\psi_1} &\leq C\sqrt{N}\|\vct{w}_k\|_2|\<\hat{\vct{b}}_\ell,\hh\>\<\hh,\bk\>| \\
	&\leq C\frac{\mu_{p-1}\mu_h^2\sqrt{N}}{L^{3/2}}.
\end{align*}
For the variance terms, we calculate
\revise{
\begin{align*}
	\sum_{k\in\Gamma_p}\E[\vct{z}_k^*\vct{z}_k] &=
	\sum_{k\in\Gamma_p}|\<\hh,\hat{\vct{b}}_\ell\>|^2|\<\bk,\hh\>|^2\vct{w}_k^*\E[(\ck\ck^*-\I)^2]\vct{w}_k \\
	&= N\sum_{k\in\Gamma_p} |\<\hh,\hat{\vct{b}}_\ell\>|^2|\<\bk,\hh\>|^2\|\vct{w}_k\|^2_2
	\qquad\text{(by Lemma~\ref{lm:EckckI2})} \\
	&\leq\frac{\mu_{p-1}^2\mu_h^2N}{L^2}\sum_{k\in\Gamma_p}|\<\bk,\hh\>|^2 \\
	&=\frac{\mu_{p-1}^2\mu_h^2NQ}{L^3},
\end{align*}
and
\begin{align*}
	\left\|\sum_{k\in\Gamma_p}\E[\vct{z}_k\vct{z}_k^*]\right\| &=
	\left\|\sum_{k\in\Gamma_p}|\<\hh,\hat{\vct{b}}_\ell\>|^2|\<\bk,\hh\>|^2
	\E[(\ck\ck^*-\I)\vct{w}_k\vct{w}_k^*(\ck\ck^*-\I)] \right\| \\
	&= \left\|\sum_{k\in\Gamma_p}|\<\hh,\hat{\vct{b}}_\ell\>|^2|\<\bk,\hh\>|^2\|\vct{w}_k\|^2_2\I\right\|
	\qquad\text{(by Lemma~\ref{lm:Eckckv})} \\
	&\leq \frac{\mu_{p-1}^2\mu_h^2}{L^2}\sum_{k\in\Gamma_p}|\<\bk,\hh\>|^2 \\
	&= \frac{\mu_{p-1}^2\mu_h^2Q}{L^3}.
\end{align*}
}
Thus
\begin{align}
	\label{eq:mupvectorbound}
	\left\|\sum_{k\in\Gamma_p}\vct{z}_k\right\|_2 &\leq
	C_\alpha\,\frac{\mu_{p-1}\mu_h\sqrt{N\log L}}{L^{3/2}}
	\max\left\{\sqrt{Q},\mu_h\log(\mu_h)\sqrt{\log L}\right\}
\end{align}
with probability at least $1-L^{-\alpha}$.

Combining \eqref{eq:mupscalarbound} and \eqref{eq:mupvectorbound} and taking the union bound over all $\ell\in\Gamma_{p+1}$ yields
\begin{align*}
	\mu_p &\leq \mu_{p-1}\,\frac{C_\alpha\sqrt{MQ\log L}}{L}, 
\end{align*}
with probability exceeding $1-2QL^{-\alpha}$.  Then taking $Q$ as in \eqref{eq:Qbound} and the union bound over $1\leq p\leq P$ establishes the lemma.
%
\subsection{Proof of Lemma~\ref{lm:ApWpnorm}}

\revise{
We start by fixing $p$ and again writing 
\[
	\left\|\cA_p^*\cA_p\mtx{W}_{p-1}-\frac{Q}{L}\mtx{W}_{p-1}\right\| ~=~
	\left\|\cA_p^*\cA_p\mtx{W}_{p-1}- \E[\cA_p^*\cA_p\W_{p-1}]\right\|.
\]
We can rewrite this as the spectral norm of a sum of random rank-1 matrices:
}
\begin{equation}
	\label{eq:sumZk}
	\cA_p^*\cA_p\W_{p-1} - \E[\cA_p^*\cA_p\W_{p-1}] = 
	\sum_{k\in\Gamma_p} \Zk, \qquad \Zk :=\bk\bk^*\W_{p-1}(\ck\ck^*- \I).
\end{equation}
We will use Proposition~\ref{prop:matbernpsi} to estimate the size of this random sum; we proceed by calculating the key quantities involved.  With $\vct{w}_k = \W_{p-1}^*\bk$, we can bound the size of each term in the sum as
\begin{align*}
	\|\Zk\| &= \|\bk\bk^*\W_p(\ck\ck^*-\I)\| \\
	&= \|\bk\|_2\, \|(\ck\ck^*-\I)\vct{w}_k\|_2 \\
	&\leq \mu_{\max}\sqrt{\frac{K}{L}}\,  \|(\ck\ck^*-\I)\vct{w}_k\|_2
\end{align*}
and then applying Lemmas~\ref{lm:ckmcktail} and \ref{lm:tailtopsi1} yields
\revise{
\begin{align*}
	\|\Zk\|_{\psi_1} &\leq C\,\mu_{\max}\sqrt{\frac{KN}{L}}\,\|\vct{w}_k\|_2 
	~\leq~ C\,\mu_{\max}\mu_p\frac{\sqrt{KN}}{L}.
\end{align*}
For the variance terms, we calculate
\begin{align*}
	\left\|\sum_{k\in\Gamma_p}\E[\Zk^*\Zk]\right\| &=
	\left\|\sum_{k\in\Gamma_p}\|\bk\|^2_2\E[(\ck\ck^*-\I)\vct{w}_k\vct{w}_k^*(\ck\ck^*-\I)]\right\| \\
	&=\sum_{k\in\Gamma_p}\|\bk\|^2_2\|\vct{w}_k\|^2_2
	\qquad\text{(by Lemma~\ref{lm:Eckckv})} \\
	&\leq \frac{\mu_{\max}^2K}{L}\sum_{k\in\Gamma_p}\|\W_p^*\bk\|^2_2 \\
	&= \frac{\mu_{\max}^2KQ}{L^2}\|\W_p\|_F^2
	\qquad\text{(using \eqref{eq:Bportho})},
\end{align*}
and
\begin{align*}
	\left\|\sum_{k\in\Gamma_p}\E[\Zk\Zk^*]\right\| &=
	\left\|\sum_{k\in\Gamma_p}\bk\vct{w}_k^*\E[(\ck\ck^*-\I)^2]\vct{w}_k\bk^*\right\| \\
	&= N\left\|\sum_{k\in\Gamma_p}\|\vct{w}_k\|^2_2\bk\bk^*\right\|
	\qquad\text{(by Lemma~\ref{lm:EckckI2})} \\
	&\leq \frac{\mu_p^2N}{L}\left\|\sum_{k\in\Gamma_p}\bk\bk^*\right\| \\
	&= \frac{\mu_p^2NQ}{L^2}.
\end{align*}
Then with $M = \max\left\{\mu_{\max}^2K,\mu_h^2 N\right\}$, we apply Proposition~\ref{prop:matbernpsi} with $t=\alpha\log L$ to get
\begin{align*}
	\|\cA_p^*\cA_p\W_{p-1} - \E[\cA_p^*\cA_p\W_{p-1}]\| &\leq
	C_\alpha\,2^{-p}\,\frac{\sqrt{M\log L}}{L}\max\left\{\sqrt{Q},\sqrt{M\log L}\log(M)\right\},
\end{align*}
with probability exceeding $1-L^{-\alpha}$.  With $Q$ as in \eqref{eq:Qbound}, this becomes
\begin{align*}
	\|\cA_p^*\cA_p\W_{p-1} - \E[\cA_p^*\cA_p\W_{p-1}]\| &\leq
	C_\alpha\,2^{-p}\,\frac{Q}{L}\max\left\{\frac{1}{\sqrt{C'_\alpha\log M}},\frac{1}{C'_\alpha}\right\}  \\
	&\leq 2^{-p}\frac{3Q}{4L},
\end{align*}
for an appropriate choice of $C'_\alpha$.  Applying the union bound over all $p=1,\ldots,P$ establishes the lemma.
}

\section{Supporting Lemmas}
\label{sec:supportinglemmas}

\begin{lem}
	Let $\ck\in\C^N$ be normally distributed as in \eqref{eq:normalck}, and let $\vct{u}\in\C^N$ be an arbitrary vector.  Then $|\<\ck,\vct{u}\>|^2$ is a chi-square random variable with two degrees of freedom and
	\[
		\P{|\<\ck,\vct{u}\>|^2 > \lambda} ~\leq~ \er^{-\lambda/\|\vct{u}\|^2_2}.
	\]
\end{lem}


\begin{lem}
	\label{lm:maxck1tail}
	Let $\ck\in\C^N$ be normally distributed as in \eqref{eq:normalck}.  Then
	\begin{equation}
		\label{eq:cknorm}
		\P{\|\ck\|^2_2 > Nu} ~\leq~1.2\,\er^{-u/8},\qquad \text{for all $u\geq 0$},
	\end{equation}
	and since $1.2\er^{-1/8N}\geq 1$ for all $N\geq 1$,
	\begin{equation*}
		\P{\max(\|\ck\|^2_2,1) > Nu} ~\leq~ 1.2\,\er^{-u/8}.
	\end{equation*}
\end{lem}
\begin{proof}
	It is well-known (see, for example, \cite{dasgupta03el}) that
	\begin{equation}
		\label{eq:gaussiannorm}
		\P{\|\ck\|^2_2 > N(1+\lambda)} ~\leq~
		\begin{cases}
			\er^{-\lambda^2/8} & 0\leq\lambda\leq 1 \\
			\er^{-\lambda/8} & \lambda \geq 1
		\end{cases}
		~\leq~
		1.05\, \er^{-\lambda/8},~~\lambda\geq 0.
	\end{equation}
	Plugging in $\lambda=u-1$ above yields
	\[
		\P{\|\ck\|^2_2 > Nu} ~\leq~ 1.2\,\er^{-u/8},\quad u\geq 1.
	\]
	Since $1.2\,\er^{-1/8} > 1$, the bound above can be extended for all $u\geq 0$.
\end{proof}

\begin{lem}
	\label{lm:EckckI2}
	Let $\ck\in\C^N$ be normally distributed as in \eqref{eq:normalck}.  Then
	\[
		\E[(\ck\ck^*-\I)^2] = N\I.
	\]
\end{lem}
\begin{proof}
	Using the expansion
	\[
		(\ck\ck^*-\I)^2 = \|\ck\|^2_2\ck\ck^* - 2\ck\ck^* + \I,
	\]
	we see that the only non-trivial term is $\mtx{R} = \|\ck\|^2_2\ck\ck^*$.  We compute the expectation of an entry in this matrix as
	\[
		\E[R(i,j)] = \sum_{n=1}^N\E[|\hat{c}_k[n]|^2\hat{c}_k[i]\hat{c}_k[j]^*]=
		\begin{cases}
			\sum_n\E[|\hat{c}_k[n]|^2|\hat{c}_k[i]|^2] & i=j \\
			0 & i\not= j
		\end{cases}.
	\]
	For the addends in the diagonal term
	\[
		\E[|\hat{c}_k[n]|^2|\hat{c}_k[i]|^2] = 
		\begin{cases}
			\E[|\hat{c}_k[n]|^4] = 2 & n=i \\
			1 & n\not= i
		\end{cases},
	\]
	where the calculation for $n=i$ relies on the fact that $\E[|\hat{c}_k[n]|^4]$ is the second moment of a chi-square random variable with two degrees of freedom.  Thus $\E[\mtx{R}] = (N+1)\I$, and
	\[
		\E[(\ck\ck^*-\I)^2] = (N+1)\I - 2\I + \I = N\I.
	\]
\end{proof}

\begin{lem}
	\label{lm:ckv2subexp}
	Let $\ck\in\C^N$ be normally distributed as in \eqref{eq:normalck}, and let $\vv$ be an arbitrary vector.  Then $\E[|\<\ck,\vv\>|^2]=\|\vv\|_2^2$ and
	\[
		\P{\left||\<\ck,\vv\>|^2-\|\vv\|_2^2\right| > \lambda} ~\leq~
		2.1\,\exp\left(-\frac{\lambda}{8\|\vv\|^2_2}\right).
	\]
\end{lem}
\begin{proof}
	A slight variation of \eqref{eq:gaussiannorm} gives us that
	\[
		\P{\left||\<\ck,\vv\>|^2-\|\vv\|^2_2\right| > \lambda} ~\leq~
		\begin{cases}
			2\er^{-\lambda^2/8\|\vv\|^2_2} & 0\leq\lambda\leq 1 \\
			\er^{-\lambda/8\|\vv\|^2_2} & \lambda > 1
		\end{cases}.
	\]
	The lemma follows from combining these two cases into one subexponential bound.
\end{proof}

\begin{lem}
	\label{lm:ckmcktail}
	Let $\ck\in\C^N$ be normally distributed as in \eqref{eq:normalck}, and let $\vv\in\C^N$ be an arbitrary vector.  Then
	\[
		\P{\|(\ck\ck^*-\I)\vv\|_2 > \lambda} ~\leq~ 3\exp\left(-\frac{\lambda}{\sqrt{8N}\|v\|_2}\right).
	\]
\end{lem}
\begin{proof}
	We have
	\[
		\|(\ck\ck^*-\I)\vv\|_2 = \|\<\vv,\ck\>\ck - \vv\|_2 \leq |\<\vv,\ck\>\|\ck\|_2 + \|\vv\|_2.
	\]
	For the first term above, we have for any $\tau > 0$, 
	\begin{align*}
		\P{|\<\vv,\ck\>|\,\|\ck\|_2 > \lambda\sqrt{N}\|\vv\|_2} &\leq
		\P{|\<\vv,\ck\>| > \sqrt{\lambda}\|\vv\|_2/\tau} + \P{\|\ck\|_2 > \tau\sqrt{\lambda N}} \\
		&= \P{|\<\vv,\ck\>|^2 > \lambda\|\vv\|_2^2/\tau^2} + \P{\|\ck\|_2^2 > \tau^2\lambda N}
	\end{align*}
	We can then use the fact that $|\<\vv,\ck\>|^2$ is a chi-squared random variable along with \eqref{eq:cknorm} above to derive the following tail bound:
	\begin{align*}
		\P{|\<\vv,\ck\>|\,\|\ck\|_2 > \lambda\sqrt{N}\|\vv\|_2} &\leq
		\er^{-\lambda/\tau^2} + 1.05\,\er^{-\tau^2\lambda/8} \\
		&= 2.05\, \er^{-\lambda/\sqrt{8}},
	\end{align*}
	where we have chosen $\tau^2=\sqrt{8}$.  Thus
	\begin{align*}
		\P{|\<\vv,\ck\>|\,\|\ck\|_2 + \|\vv\|_2 > \lambda} &\leq
		2.05\,\er^{1/\sqrt{8}}\cdot\er^{-\lambda/\sqrt{8N}}.
	\end{align*}
\end{proof}

\begin{lem}
	\label{lm:Eckckv}
	Let $\ck\in\C^N$ be normally distributed as in \eqref{eq:normalck}, and let $\vv\in\C^N$ be an arbitrary vector.  Then
	\[
		\E[(\ck\ck^*-\I)\vv\vv^*(\ck\ck^*-\I)] = \|\vv\|_2^2\I.
	\]
\end{lem}
\begin{proof}
	We have
	\begin{align*}
		\E[(\ck\ck^*-\I)\vv\vv^*(\ck\ck^*-\I)] &=
		\E[|\<\vv,\ck\>|^2\ck\ck^* - \ck\ck^*\vv\vv^* - \vv\vv^*\ck\ck^* - \vv\vv^*] \\
		&= \E[|\<\vv,\ck\>|^2\ck\ck^*] - \vv\vv^*.
	\end{align*}
	Let $R(i,j)$ be the entries of the first matrix above:
	\begin{align*}
		R(i,j) &= \E[|\<\vv,\ck\>|^2\hat{c}_k[i]\hat{c}_k[j]^*] \\
		&= \sum_{n_1,n_2}v[n_1]v[n_2]\E[\hat{c}_k[n_1]\hat{c}_k[n_2]^*\hat{c}_k[i]\hat{c}_k[j]^*].
	\end{align*}
	On the diagonal, where $i=j$, all of the terms in the sum above are zero except when $n_1=n_2$, and so
	\[
		R(i,i) = \sum_{n=1}^N |v[n]|^2 \E\left[|\hat{c}_k[n]|^2|\hat{c}_k[i]|^2\right].
	\]
	Using the fact that
	\[
		\E\left[|\hat{c}_k[n]|^2|\hat{c}_k[i]|^2\right] =
		\begin{cases}
			2 & n=i \\
			1 & n\not= i
		\end{cases},
	\]
	we see that $R(i,i) = |v[i]|^2 + \|\vv\|_2^2$.  Off the diagonal, where $i\not= j$, we see immediately that $\E[\hat{c}_k[n_1]\hat{c}_k[n_2]^*\hat{c}_k[i]\hat{c}_k[j]^*]$ will be zero unless one of two (non-overlapping) conditions hold: $(n_1=i,n_2=j)$ or $(n_1=j,n_2=i)$.  Thus
	\begin{align*}
		R(i,j) &= v[i]v[j]\E[\hat{c}_k[i]^2]\E[\hat{c}_k[j]^2] + v[j]v[i]\E[|\hat{c}_k[j]|^2]\E[|\hat{c}_k[i]|^2].
	\end{align*}
	Note the lack of absolute values in the first term on the right above; in fact, since the $\hat{c}_k[i]$ have uniformly distributed phase, $\E[\hat{c}_k[i]^2]=\E[\hat{c}_k[j]^2]=0$, and so $R(i,j) = v[i]v[j]$.  As such
	\[
		\E[(\ck\ck^*-\I)\vv\vv^*(\ck\ck^*-\I)] = \E[|\<\vv,\ck\>|^2\ck\ck^*] - \vv\vv^* = 
		\vv\vv^* + \|\vv\|_2^2\I - \vv\vv^* = \I.
	\]
\end{proof}

\begin{lem}
	\label{lm:ckvuck}
	Let $\ck\in\C^N$ be normally distributed as in \eqref{eq:normalck}, and let $\uu,\vv\in\C^N$ be arbitrary vectors.  Then
	\begin{align*}
		\P{|\<\ck,\vct{v}\>\<\vct{u},\ck\> - \<\vct{u},\vct{v}\>| > \lambda} &\leq
		2\er\cdot\exp\left(-\frac{\lambda}{\|\vct{u}\|_2\|\vct{v}\|_2}\right).
	\end{align*}
\end{lem}
\begin{proof}
	For any $t > 0$,
	\begin{align*}
		\P{|\<\ck,\vct{v}\>\<\vct{u},\ck\>| > \lambda} &\leq
		\P{|\<\ck,\vct{v}\>| > t} + \P{|\<\vct{u},\ck\>| > \lambda/t} \\
		&= \P{|\<\ck,\vct{v}\>|^2 > t^2} + \P{|\<\vct{u},\ck\>|^2 > \lambda^2/t^2} \\
		&\leq \exp\left(-\frac{t^2}{\|\vct{v}\|_2^2}\right) +
		\exp\left(-\frac{\lambda^2}{t^2\|\vct{u}\|^2_2}\right).
	\end{align*} 
	Choosing $t^2=\lambda\|\vct{v}\|_2/\|\vct{u}\|_2$ yields
	\begin{align*}
		\P{|\<\ck,\vct{v}\>\<\vct{u},\ck\>| > \lambda} &\leq
		2\exp\left(-\frac{\lambda}{\|\vct{u}\|_2\|\vct{v}\|_2}\right),
	\end{align*}
	and so
	\begin{align*}
		\P{|\<\ck,\vct{v}\>\<\vct{u},\ck\> - \<\vct{u},\vct{v}\>| > \lambda} &\leq
		\P{|\<\ck,\vct{v}\>\<\vct{u},\ck\>| > \lambda - \|\vct{u}\|_2\|\vct{v}\|_2} \\
		&\leq 2\exp\left(-\frac{\lambda}{\|\vct{u}\|_2\|\vct{v}\|_2} + 1\right) \\
		&= 2\er\cdot\exp\left(-\frac{\lambda}{\|\vct{u}\|_2\|\vct{v}\|_2}\right).
	\end{align*}
\end{proof}


\bibliographystyle{IEEEtran}
\bibliography{bd-references}


\end{document}